\documentclass[a4paper,reqno]{amsart}
\usepackage[english]{babel}
\usepackage[T1]{fontenc}
\usepackage{amsmath,amssymb,amsthm}
\usepackage{hyperref,todonotes}
\usepackage{graphicx}
\usepackage[bbgreekl]{mathbbol}
\usepackage{mathtools}
\usepackage[all,cmtip]{xy}
\usepackage{tikz-cd}
\usepackage{csquotes}

\newtheorem{definition}{Definition}
\newtheorem{theorem}[definition]{Theorem}
\newtheorem{proposition}[definition]{Proposition}
\newtheorem{lemma}[definition]{Lemma}
\newtheorem{corollary}[definition]{Corollary}
\newtheorem{remark}[definition]{Remark}

\usepackage[bibstyle=alphabetic,citestyle=alphabetic,useprefix,giveninits=true, minbibnames=99,maxbibnames=99,backend=biber]{biblatex}
\renewbibmacro{in:}{} 

\DeclareSourcemap{
  \maps[datatype=bibtex]{
    \map[overwrite]{ 
      \step[fieldsource=doi, final]
      \step[fieldset=url, null]
      \step[fieldset=eprint, null]
      }
     \map[overwrite]{ 
      \step[fieldsource=eprint, final]
      \step[fieldset=pages, null]
      \step[fieldset=eid, null]
      \step[fieldset=journal, null]
    }  
  }
}

\bibliography{bibliography}

\newcommand{\X}[1]{(X_{#1})}

\newcommand{\qsp}[2]{\,\ensuremath{\raise.5ex\hbox{$#1$}\big\slash\raise-.5ex\hbox{$#2$}}} 

\newcommand{\pard}[2]{\frac{\delta#1}{\delta#2}}

\newcommand{\txi}[1]{\widetilde{\xi}^{#1}}

\newcommand{\intl}{\int\limits}
\newcommand{\trintl}[1]{\mathrm{Tr}\intl_{#1}}
\newcommand{\tc}{\widetilde{c}}
\newcommand{\tom}{\widetilde{\omega}}

\newcommand{\te}{\widetilde{e}}

\newcommand{\ttc}{\widetilde{\tc}}
\newcommand{\tte}{\widetilde{\te}}
\newcommand{\ttom}{\widetilde{\tom}}
\newcommand{\ttxi}[1]{\widetilde{\widetilde{\xi^{#1}}}}
\newcommand{\paral}{\slash\!\slash}
\newcommand{\oloc}{\Omega_{\mathrm{loc}}}

\newcommand{\filt}[1]{(#1)}
\newcommand{\filtint}[1]{M^{(#1)}}
\newcommand{\filtBF}[1]{(#1 )}
\newcommand{\filtintBF}[1]{M^{(#1 )}}

\newcommand{\comp}[1]{\langle #1 \rangle}

\newcommand{\BFnd}{BF_{*}}

\title[BV-BFV 3d General Relativity]{Fully extended BV-BFV description of General Relativity in three dimensions}

\author{G. Canepa}
\address{Institut f\"ur Mathematik, Universit\"at Z\"urich, Winterthurerstrasse 190, 8057 Z\"urich, Switzerland}
\email{giovanni.canepa@math.uzh.ch}

\author{M. Schiavina}
\address{Institute for Theoretical Physics, ETH Zurich, Wolfgang Pauli strasse 27, 8093, Z\"urich, Switzerland}
\address{Department of Mathematics, ETH Zurich, R\"amistrasse 101, 8092, Z\"urich, Switzerland }
\email{micschia@phys.ethz.ch}

\thanks{G.C. acknowledges partial support of SNF Grant No. 200020 172498/1. His research was (partly) supported by the NCCR SwissMAP, funded by the Swiss National Science Foundation, and by the COST Action MP1405 QSPACE, supported by COST (European Cooperation in Science and Technology). 
M.S. acknowledges partial support from Swiss National Science Foundation grants P2ZHP2\_164999 and P300P2\_177862.}

\begin{document}

\begin{abstract}
We compute the extension of the BV theory for three-dimensional General Relativity to all higher-codimension strata - boundaries, corners and vertices - in the BV-BFV framework. Moreover, we show that such extension is strongly equivalent to (nondegenerate) BF theory at all codimensions.
\end{abstract}

\maketitle
\tableofcontents

\section*{Introduction}

The BV-BFV formalism is a combination of a Lagrangian approach to field theories with on-shell symmetries --- the BV formalism --- and of its counterpart for constrained Hamiltonian systems --- the BFV formalism --- named after the work of Batalin, Fradkin and Vilkovisky \cite{BV1,BV2,BV3}. 

The link between the two approaches was developed by Cattaneo, Mnev and Reshetikhin in \cite{CMR} as a first step towards quantisation of gauge theories on manifolds with boundaries, with an axiomatisation of Classical and Quantum field theory in mind. The main idea is to allow boundaries to spoil gauge invariance of a given theory, described by its BV-cohomology, but to control such failure by means of induced cohomological data associated to the boundary. This interplay between bulk and boundary data allows for a consistent quantisation scheme \cite{CMR2},  compatible with gluing by construction, whose output is the cohomology of a quantum operator that directly encodes gauge invariance.

The program of approaching General Relativity (GR) from this point of view was initiated in \cite{ScTH}, showing that when diffeomorphism symmetry is involved, the induction procedure essential for the BV-BFV correspondence to hold is far from being guaranteed. Indeed, while the standard example of General Relativity in the Einstein--Hilbert formulation satisfies all the axioms, thus providing a well defined BV-BFV theory in all spacetime dimensions $d\not=2$ \cite{CSEH}, obstructions arise in 4d General Relativity in the (Einstein--Sciama--Kibble--) Palatini--Cartan--Holst tetradic formulation\footnote{The physics literature for this version of GR seems to disagree on standard nomenclature. This is also why we refrain from naming a particular version of the theory in the title of the current paper.} \cite{CSPCH}. Similar obstructions have been found in one-dimensional reparametrisation models \cite{CStime}, or are to be expected in certain supersymmetric models \cite{GP}.

In this paper we show that General Relativity {in spacetime dimension $d=3$}, phrased in the triadic language \cite{W2, Carlip, Cartan, Wise}, admits an extension to all higher codimension strata in the BV-BFV sense (Definition \ref{def:n-extended-BVBFV}){, in stark contrast with the $d\geq4$ analogue, where the procedure fails at codimension 1 \cite{CS2017}}. This means not only that we can control the failure of gauge invariance of the theory upon introducing a boundary, but that the boundary gauge invariance is controlled by corner cohomology, and so on. One says in this case that 3d GR is a fully extended BV-BFV theory. 

We obtain explicit expressions for the BV-BFV data at higher codimension strata, {making 3d GR in triad variables} the first example of a fully extended theory that features a nontrivial symplectic reduction at every step. 
{ We stress that in the Einstein--Hilbert formulation in dimension $d>2$, the extension to corners and higher codimension strata has not been accomplished.}

Out of such data, one can directly read relevant information such as the algebra of constraints together with a cohomological presentation of the reduced phase space (from the codimension-$1$ data), and the representations carried by \emph{boundary insertions} (in codimension-$2$). Moreover, following \cite{MSW2019}, a fully extended BV-BFV theory induces a solution of Witten descent equations \cite{W1} - a step towards the understanding of observables in General Relativity - and one can discuss the emergence of edge modes and holographic counterparts (see, e.g. \cite{CHvD}){, as well as asymptotic symmetries \cite{RS}. Investigations in this direction will be carried out in a future paper.}

Furthermore, { we show that the fully extended BV-BFV description of GR is strongly equivalent to that of \emph{nondegenerate} BF theory, by extending to higher codimensions the notion of strong equivalence of BV theories and the explicit equivalence found in \cite{CSS2017}.} This means that, on every stratum, the data of GR { and that of BF theory differ by a change of coordinates at every codimension.}

Although it is well known that 3d General Relativity is classically equivalent to nondegenerate BF theory \cite{W2}, in this paper we explicitly write down the symmetries in terms of diffeomorphisms, and show that this description is equivalent to the one coming from standard symmetries of BF theory, at all codimensions. This is the result of nontrivial calculations, interesting also as they can be taken as a guideline or bootcamp for the more involved case of 4d gravity.

{One interpretation of this result suggests thinking of BF theory as a (possibly degenerate) extension of GR (for $d=3$), coinciding on an \emph{open sector}, i.e. when the nondegeneracy condition is imposed. From this point of view, BV-BFV quantisation of BF theory (which has been carried out explicitly in \cite{CMR2}) may be taken as a slightly more general quantisation of GR in three dimensions, compatible with cutting and gluing along submanifolds\footnote{The adaptation of a quantisation of BF theory to include corner data has been proposed by \cite{IM2018}.}. The nondegeneracy condition can then be imposed directly at the quantum level, without spoiling the quantisation procedure. Thus, the present paper paves the way for a direct application of these extended quantisation techniques to the example of GR, bringing a model for gravity a step closer to functorial approaches to quantum field theory, by assigning compatible structure to higher-codimension strata, all the way down to points.}

In Section \ref{sec:BVBFVGRproof} we describe in detail the constructive steps one needs in order to obtain the BV-BFV data at every codimension of a stratified manifold $\{M^{\filt{k}}\}_{k=0\dots 3}$. We divide the proof of the main Theorem, stating that GR in the BV formalism is fully extended, into three Propositions, each of which is aimed at recovering data one codimension further.

{In Section \ref{sec:BF-GR_equivalence} we present explicit symplectomorphisms between the spaces of fields $\mathcal{F}^{\filt{k}}_{GR/\BFnd}$ at every codimension, and we show how they commute with the BV-BFV surjective submersion maps. This proves that the BV-BFV induction commutes with equivalence at codimension-$k$.}

The results in this paper show how diffeomorphisms can be seen as an equivalent choice of a BV-extension of classical BF theory, and fully describe the compatibility with higher-codimension strata, completely characterising the symmetries of GR in three dimensions.

\subsection*{Acknowledgements}
We thank Alberto S. Cattaneo for interesting discussions and helpful insight. G.C. is grateful for hospitality to the Department of Mathematics of the University of California at Berkeley.

\section{Preliminaries}
The strategy employed in this paper is to consider the BV-data associated to a manifold $M$ and, step by step, analyse what structure it induces if we allow $M$ to carry a stratification of increasing codimension. We follow here
the adaptation of the classical BV-BFV axioms introduced in \cite{CMR2012}, as proposed by \cite{MSW2019}. 

{
\begin{definition}
Let $M$ be an $m$-dimensional smooth manifold. A $n$-stratification of $M$ is a filtration of smooth manifolds (possibly with boundary) $\{ M ^{\filt{k}}\}_{k=0\dots n}$ of dimension $\mathrm{dim}(M^{\filt{k}})=m-k$, with $\filtint{0}=M$, such that there exists a smooth embedding $\iota^{\filt{k+1}}\colon M^{\filt{k+1}} \to  M^{\filt{k}}$ for every $0\leq k < n$.
\end{definition}
}

\begin{remark}
A particular example of a stratification is given by a manifold with corners (and vertices, i.e. boundaries of corners), where the connected components of boundaries, corners and vertices compose the cells of a stratum $M^{\filt{k}}$. {For practical purposes the reader can consider this example as the main application, although our definitions allow for more general embedded submanifolds.}
\end{remark}

{
\begin{definition}[\cite{MSW2019}]\label{def:n-extended-BVBFV}
A \emph{strict, $n$-extended, exact BV-BFV theory}, shorthanded with \emph{$n$-extended theory}, is the assignment, to the $n$-stratification $\{M^{(k)}\}_{k=0\dots n}$, of the data 
	$$\mathfrak{F}^{\uparrow n}=(\mathcal{F}^{(k)}, S^{(k)}, \alpha^{(k)}, Q^{(k)},\pi^{(k)})_{k=0\dots n},$$
such that, for every $0\leq k \leq n$,
\begin{enumerate}
\item\label{symplecticcondition} $\mathcal{F}^{(k)}$ is the space of sections of a graded vector bundle $E^{(k)}\longrightarrow M^{(k)}$ and $\alpha^{(k)}\in\oloc^{1}(\mathcal{F}^{(k)})$ is a degree-$(k-1)$ local form\footnote{In this paper local forms are differential forms on $\mathcal{F}^{\filt{k}}$ which depend only on a finite number of derivatives of fields, i.e. sections of $E$. See \cite{MSW2019}.}, such that $\varpi^{(k)}=\delta\alpha^{(k)}$ is weakly symplectic on $\mathcal{F}^{(k)}$, and $\delta$ is the de Rham differential on $\mathcal{F}^{\filt{k}}$,
\item $\pi^{(k)}\colon \mathcal{F}^{(k)} \longrightarrow \mathcal{F}^{(k+1)}$ is a surjective submersion\footnote{For $k=n$ the projection $\pi^{\filt{n}}$ is the unique map from $\mathcal{F}^{\filt{n}}$ to the empty set.},
\item $Q^{(k)}$ is an evolutionary, cohomological, odd vector field of degree $1$ on $\mathcal{F}^{(k)}$, i.e. $[\mathcal{L}_{Q^{(k)}},d]=[Q^{(k)},Q^{(k)}]=0$, that is also projectable: $Q^{(k+1)} = (\pi^{(k)})_* Q^{(k)}$,
\item $S^{(k)}\in \Omega_{\mathrm{loc}}^{0}(\mathcal{F}^{(k)})$ is a (real-valued) degree-$k$ local functional,
\end{enumerate}
such that, for $0\leq k\leq n-1$,
\begin{align} \label{rCME}
		  \iota_{Q^{(k)}} \varpi^{(k)} = \delta S^{(k)} + \pi^{(k)*}\alpha^{(k + 1)}
\end{align}
whereas for $k=n$, we require
\begin{equation}\label{terminalCME}
\iota_{Q^{\filt{n}}} \varpi^{\filt{n}} = \delta S^{\filt{n}}. 
\end{equation}
When $n=\mathrm{dim}(M^{(0)})$ we say that the theory is \emph{fully extended}. When $n=0$, the data defines a \emph{BV theory}.
\end{definition}

\begin{remark}
Notice that a direct consequence of \eqref{rCME} is
\begin{equation}\label{b-action}
		\frac12 \iota_{Q^{(k)}}\iota_{Q^{(k)}}\varpi^{(k)} = \pi^{(k)*} S^{(k+1)},
\end{equation}
as well as  $\iota_{Q^{\filt{n}}} \iota_{Q^{\filt{n}}} \varpi^{\filt{n}}= 0$. This is proven in \cite[Proposition 3.1]{CMR2012}.
\end{remark}
}

\begin{remark}
A direct interpretation of Equation \eqref{rCME} is that, on every stratum, $Q^{\filt{k}}$ is the Hamiltonian vector field of the action functional $S^{\filt{k}}$ up to boundary terms \eqref{rCME}, and that the classical master equation is satisfied up to boundary terms \eqref{b-action}. 
\end{remark}

We will be concerned here with classical field theories that enjoy symmetries given by Lie algebra actions, and our starting point to build an extended (exact) BV-BFV theory is a couple $(\mathcal{F}^{cl}, S^{cl})$. The space of classical fields $\mathcal{F}^{cl}$ is the space of sections\footnote{In the present paper we will consider principal connections as fields. However, we can reduce to this setting by expanding around an arbitrary reference connection.} of some sheaf or bundle  $E \rightarrow M$, while $S^{cl}$ is a local functional on $\mathcal{F}^{cl}$, i.e. a function of the fields and a finite number of jets, called \emph{action functional}. The symmetry data is sometimes encoded in an involutive distribution\footnote{In full generality the BV formalism only requires that $\mathcal{D}^{cl}$ be involutive on the critical locus of $S^{cl}$, i.e. the space of solutions of the associated variational problem.} $\mathcal{D}^{cl}$ on $\mathcal{F}^{cl}$. 

The first step is to build a BV theory (or 0-extended BV-BFV) $\mathfrak{F}=\left(\mathcal{F},S,\varpi,Q\right)$, {is to promote the space of classical field to a $(-1)$-symplectic graded manifold $(\mathcal{F},\varpi)$, and to construct a cohomological vector field $Q$ on $\mathcal{F}$, the Hamiltonian vector field of $S$ w.r.t. $\varpi$ \cite{BV1,BV2}.} It is then possible to build an $m$-extended BV-BFV theory using a constructive approach. 

Let  $\{ M ^{\filt{k}}\}_{k=0\dots n}$ be an $n$-stratification of $M$ and consider on it an $n$-extended BV-BFV theory. According to Definition \ref{def:n-extended-BVBFV} on the $n$-th stratum we have equation \eqref{terminalCME}. If we allow an $(n+1)$-codimension stratum, Equation \eqref{terminalCME} will likely be spoiled, or we simply extend the theory by zero. In the former case, if we can find $\pi^{\filt{n}}\colon \mathcal{F}^{\filt{n}}\longrightarrow \mathcal{F}^{\filt{n+1}}$, together with $\alpha^{\filt{n+1}}$ and $S^{\filt{n+1}}$ satisfying \eqref{rCME} (and \eqref{b-action}), and $\varpi^{\filt{n+1}}= \delta \alpha^{\filt{n+1}}$ non degenerate, we will have extended the BV-BFV theory to codimension-$(n+1)$. 

In a practical scenario, this goes through by integrating by parts the terms in $\delta S^{\filt{n}}$, but the resulting data on the higher-codimension stratum does not automatically satisfy the axioms in Definition \ref{def:n-extended-BVBFV}. {In particular, the existence of $\mathcal{F}^{\filt{n+1}}$ as a smooth symplectic manifold (and hence of $\pi^{\filt{n}}$)} is not always guaranteed (see \cite{CSPCH,CStime}). When this happens the theory is then only \emph{$n$-extendable}. We summarise the previous discussion with the following definition:

\begin{definition}\label{def:extendabletheories}
Let $M$ be an $m$-dimensional smooth manifold and let $\mathfrak{F}^{\uparrow 0}$ an exact BV theory on it. We say that the BV-theory $\mathfrak{F}^{\uparrow 0}$ is $n$-extendable if, for every $n$-stratification such that $M^{\filt{0}}=M$, there exists an $n$-extended exact BV-BFV theory $\mathfrak{F}^{\uparrow n}$ associated to it. If $n=\mathrm{dim}(M)$ we will say that $\mathfrak{F}^{\uparrow 0}$ is fully extendable.
\end{definition}

\begin{remark} 
{Typically, (1-extended) BV-BFV theories are associated to a manifold with boundary.} In Definition \ref{def:n-extended-BVBFV} this amounts to taking into account a filtration where $M^{\filt{0}}=M$ is the manifold itself and $M^{\filt{1}}= \partial M$ is the boundary of the manifold. The generalization of this to higher codimensions is to consider $k$-extended BV-BFV theories on manifolds with boundary, corners, vertices, etc. and consider the filtration given by $M^{\filt{k}}$. 
\end{remark}

We give here the notion of a \emph{strong}\footnote{\label{fnequivalence}The natural notion of equivalence, given the cohomological context in which physical data is presented, would coincide with weak equivalences of BV-BFV complexes  {that preserve the BV classes of the symplectic structure and the action functional (see \cite{MSW2019})}. The definition we propose here is essentially that of an isomorphism of complexes, hence a stronger requirement.} equivalence of (extended) BV-BFV theories. It implies the standard notion of equivalence of classical field theories, which requires the critical loci of two action functionals  to be isomorphic (modulo symmetries).\footnote{{Notice that this notion is strictly weaker than the one discussed in footnote \ref{fnequivalence}, since it only requires the cohomology in degree zero to coincide.}}  In the case of BV theories the following definition has been proposed in \cite{CSS2017}:
\begin{definition}\label{def:strongBVeq}
A strong equivalence between two BV theories $\mathfrak{F}_1^{\uparrow 0}$ and $\mathfrak{F}_2^{\uparrow 0}$ is a degree-$0$ symplectomorphism
$$ \Phi : (\mathcal{F}_1^{\filt{0}}, \varpi_1^{\filt{0}} ) \rightarrow (\mathcal{F}_2^{\filt{0}}, \varpi_2^{\filt{0}})$$
preserving the BV action\footnote{We always consider action functionals $S^{\filt{0}}$ modulo constants.}: $\Phi^* S_2^{\filt{0}} = S_1^{\filt{0}}$.
\end{definition}

We can modify this definition to encompass $n$-extended BV-BFV theories.

\begin{definition}\label{def:n-ext-strongBVeq}
A strong equivalence between two $n$-extended exact BV-BFV theories $\mathfrak{F}_1^{\uparrow n}$ and $\mathfrak{F}_2^{\uparrow n}$ is a collection of symplectomorphisms
$$ \Phi^{\filt{k}} : (\mathcal{F}_1^{\filt{k}}, \varpi_1^{\filt{k}} ) \rightarrow (\mathcal{F}_2^{\filt{k}}, \varpi_2^{\filt{k}})$$
preserving the $k^{th}$ BFV action: $\Phi^* S_2^{\filt{k}} = S_1^{\filt{k}}$ and satisfying, for $0\leq k\leq n-1$
$$ \pi_2^{\filt{k}} \circ \Phi^{\filt{k}} = \Phi^{\filt{k+1}} \circ \pi_1^{\filt{k}} .$$
\end{definition}

\section{{Three-dimensional General Relativity and BF theory}}
The common framework shared by 3 dimensional General Relativity and BF theory is as follows. Let $P \rightarrow M$ be an $SO(2,1)$-principal bundle on a 3-dimensional, compact, orientable\footnote{Extensions to noncompact manifolds are possible, but outside the main objective of this paper. {See \cite{RS} for an adaptation of the BV-BFV method to boundaries at infinity.} Orientability is not necessary, but we restrict to orientable manifolds for simplicity.} smooth manifold $M$. Let also $\mathcal{V}$ be the associated vector bundle where each fibre is isomorphic to $(V, \eta)$, a 3-dimensional vector space with a pseudo-Riemannian inner product $\eta$ on it{, with normal form $\eta=\mathrm{diag}(-1,1,1)$ in an $\eta$-orthonormal basis $\{ v_i\}_i$ of $V$}. We further identify $\mathfrak{so}(2,1) \cong \wedge^2 \mathcal{V}$ using $\eta$ and we define a map Tr$: \wedge ^3 V \rightarrow \mathbb{R}$ given by the volume form and such that Tr$(v_i, v_j, v_k)= \epsilon_{ijk}$ { (we fix $\epsilon_{123}=1$)}. To keep the notation light we will use the shorthand 
$$
\int \mathrm{Tr}[\dots] \equiv \mathrm{Tr}\int \dots
$$
In the following subsections we specify the details proper to each theory.
\subsection{Three-dimensional BF theory} \label{sec:BF}
The fields of the theory are $B \in \Omega^1(M, \mathcal{V})$ and a { principal} connection $A \in \mathcal{A}_P$. We will think of $A$ as a connection form around the trivial connection, that is to say $A \in \Omega^1( M, \wedge^2\mathcal{V})$.

\begin{definition}\label{def:classBF}
Classical BF theory is the pair $(\mathcal{F}^{cl}_{BF}, S^{cl}_{BF})$ where $$\mathcal{F}^{cl}_{BF}=  \Omega^1(M, \mathcal{V}) \oplus\Omega^1( M, \wedge^2\mathcal{V}) $$ is the space of fields, and the action functional reads
$$S^{cl}_{BF} = \trintl{M} B\wedge F_A, $$
with $F_A\in \Omega^2( M, \wedge^2\mathcal{V})$ the curvature of the connection $A$. We can further require $B$ to be nondegenerate as a map $B\colon TM \to \mathcal{V}\stackrel{\eta}{\simeq}\wedge^2 \mathcal{V}$. Denoting by $\Omega_{nd}^1( M, \wedge^2\mathcal{V})$ the space of nondegenerate $B$'s, we will call the resulting theory \emph{nondegenerate} BF theory, and denote it with the notation $\BFnd$ where relevant.
\end{definition}
The symmetries of the theory comprise gauge trasformations, parametrized by\footnote{We will denote here the action of a symmetry by the notation $\delta_\chi$ with parameter $\chi$. This is a notation historically used to denote a $\chi$-dependent vector field acting on generators on the algebra of functions over $F^{cl}$. It will be replaced by a well-defined vector field when we pass to the BV formalism.} $\chi \in \Omega^0(M, \mathfrak{so}(2,1))\simeq \Omega^0(M, \wedge^2\mathcal{V})$
$$ \delta_\chi B \equiv [\chi,B] \qquad \qquad \delta_\chi A \equiv d_A \chi,$$
together with what is sometimes referred to as \emph{shift symmetry}, a traslation of $B$ parametrized by $\tau \in \Omega^0(M, \mathcal{V})$, $\delta_\tau B \equiv d_A \tau.$

We recall the BV version of three-dimensional $BF$ theory.
\begin{definition}
The BV-data for BF theory is given by $$\mathfrak{F}_{BF}^{\uparrow 0}=\left(\mathcal{F}_{BF},\alpha_{BF},S_{BF},Q_{BF}\right),$$ where the BV space of fields can be written as  
$$\mathcal{F}_{BF}= T^* [-1]\left(\Omega^1(M, \mathcal{V}) \oplus \mathcal{A}_P \oplus \Omega^0[1]( M, \wedge^2\mathcal{V})\oplus \Omega^0[1]( M, \wedge^1\mathcal{V})\right),
$$
and, if we arrange the fields in the following convenient way 
$$\mathcal{B}=\tau + B + A^\dag + \chi^\dag \in \Omega^\bullet(M, \mathcal{V})[1-\bullet], \qquad \mathcal{A}=\chi + A + B^\dag + \tau^\dag\in \Omega^\bullet(M, \wedge^2\mathcal{V})[1-\bullet],$$
the BV data reads\footnote{We use here the convention that only the admissible terms (i.e. the ones that are top forms) appear in the integrands.}
\begin{align*}
\alpha_{BF} &= \trintl{ M}  \mathcal{B} \wedge \delta \mathcal{A} \qquad \varpi_{BF} =   \delta \alpha_{BF} \\
S_{BF}&= \trintl{ M}\mathcal{B}\wedge \left( d \mathcal{A}+ \frac{1}{2}[\mathcal{A}, \mathcal{A}]\right)\\
Q_{BF}\mathcal{B}&= d_{\mathcal{A}} \mathcal{B} ; \qquad Q_{BF}\mathcal{A}=d \mathcal{A}+ \frac{1}{2}[\mathcal{A}, \mathcal{A}].
\end{align*}
If $B\in \Omega^1_{nd}(M, \wedge^2\mathcal{V})$ we will denote the resulting BV theory by $\mathfrak{F}^{\uparrow 0}_{\BFnd}$.
\end{definition}

{ BF theory is an example of an AKSZ theory \cite{AKSZ}. As such, it can be fully extended:}

\begin{theorem}\label{thm:BF-BV_BFV}\cite{CMR2012}
The BV theory $\mathfrak{F}^{\uparrow 0}_{BF}=(\mathcal{F}_{BF},S_{BF}, \alpha_{BF}, Q_{BF} )$ is fully extendable. The BV-BFV data of the fully extended theory $\mathfrak{F}^{\uparrow 3}_{BF}$ is given by the following expressions ($i=0 \dots 3$):
\begin{align*}
\alpha^{\filtBF{i}}_{BF} =& \trintl{\filtintBF{i}}  \mathcal{B} \wedge \delta \mathcal{A}\qquad \varpi^{\filtBF{i}}_{BF} = \delta \alpha^{\filtBF{i}}_{BF}\\
S^{\filtBF{i}}_{BF}=& \trintl{\filtintBF{i}}\mathcal{B}\wedge \left( d \mathcal{A}+ \frac{1}{2}[\mathcal{A}, \mathcal{A}]\right)\\
Q^{\filtBF{i}}_{BF}\mathcal{B}&= d_{\mathcal{A}} \mathcal{B} ; \qquad Q^{\filtBF{i}}_{BF}\mathcal{A}=d \mathcal{A}+ \frac{1}{2}[\mathcal{A}, \mathcal{A}]
\end{align*}
where we used once again the convention that only the admissible terms appear in the integrands and $\pi^{\filtBF{i}}_{BF}$ is the restriction of the fields to $\filtintBF{i+1}$.
\end{theorem}
Note that in this notation $BF$ theory is \textit{self-similar}, i.e. the action $S^{\filtBF{i}}_{BF}$, the symplectic two form $\varpi^{\filtBF{i}}_{BF}$ and the cohomological vector field $Q^{\filtBF{i}}_{BF}$ have the same expression on bulk ($0$-stratum), boundary ($1$-stratum) and every subsequent iteration.

\subsection{Three dimensional General Relativity}\label{sec:GR}
The fields are a co-frame field $ e \in \Omega_{nd}^1(M, \mathcal{V})$, also called a \emph{triad} ($nd$ stands for non degenerate), i.e. an isomorphism $e\colon TM \rightarrow \mathcal{V}$, and an $SO(2,1)$ principal connection $\omega \in \mathcal{A}_P\simeq \Omega^1( M, \wedge^2\mathcal{V})$ (again we work around the trivial connection).

\begin{definition}\label{def:classical_GR}
Classical three dimensional General Relativity (GR) is the pair $(\mathcal{F}^{cl}_{GR}, S^{cl}_{GR})$ where $$\mathcal{F}^{cl}_{GR}=  \Omega_{nd}^1(M, \mathcal{V}) \oplus\Omega^1( M, \wedge^2\mathcal{V}) $$ is the space of fields and the action functional reads
$$S^{cl}_{GR} = \trintl{M} e\wedge F_\omega .$$
\end{definition}

In order to define a BV theory extending Definition \ref{def:classical_GR}, we have to incorporate the symmetries by extending the space of fields. The classical functional $S_{GR}^{cl}$ is invariant under the action of internal gauge transformations $SO(2,1)$ and the action of spacetime diffeomorphisms. {We parametrize their associated Lie-algebra actions with two \emph{ghost} fields, $c \in \Omega^0[1](M, \wedge^2 \mathcal{V})$ and $\xi \in \Gamma [1](TM)$ respectively}:
\begin{align*}
\delta_\xi e \equiv L_\xi^\omega e \qquad \qquad \delta_\xi\omega \equiv \iota_\xi F_\omega\\
\delta_c e \equiv [c,e] \qquad \qquad \delta_c \omega \equiv d_\omega c
\end{align*}
where $L_\xi^\omega\coloneqq [\iota_\xi,d_\omega]$ is the graded commutator between the contraction with respect to $\xi$ (a degree-$0$ derivation), and $d_\omega$ is the covariant derivative (a degree-$1$ derivation).
With these quantities we can also define $\iota_{[\xi,\xi]} := [L_\xi^\omega, \iota_\xi]$. Note that by  \cite[Lemma 18]{CS2017} $\iota_{[\xi,\xi]} = [L_\xi^\omega, \iota_\xi]= [L_\xi, \iota_\xi]$.

The BV structure associated to these symmetries has been studied in generality in \cite[Section 3]{CS2017}. 

\begin{definition}\label{def:BVGR}
The BV theory for General Relativity in three dimensions is given by the data $\mathfrak{F}^{\uparrow 0}_{GR}=(\mathcal{F}_{GR},S_{GR}, \alpha_{GR}, Q_{GR} )$ where the BV space of fields is 
 $$\mathcal{F}_{GR}= T^* [-1]\left(\Omega_{nd}^1(M, \mathcal{V}) \oplus \mathcal{A}_P \oplus \Omega^0[1]( M, \wedge^2\mathcal{V})\oplus \Gamma[1]TM\right)$$
denoting the fields in the cotangent fibre by $e^\dag=\Omega^2[-1](M,\wedge^2 \mathcal{V})$, $\omega^\dag\in\Omega^2[-1](M,\mathcal{V})$, $c^\dag= \Omega^{\textrm{top}}[-2](M,\mathcal{V})$ and {$\xi^\dag\in\Omega^{\text{top}}[-2](M,T^*M)$}, and symmetry generators as $\xi\in\Gamma[1]TM$ and $c\in\Omega^0[1](M,\wedge^2\mathcal{V})$, the BV one-form and action functional are\footnote{See Remark \ref{Antighostxi} for the meaning of the trace on the last term.}
\begin{equation*}
{\alpha_{GR}=  \trintl{ M} e^\dag \delta e + \omega^\dag \delta \omega +c^\dag \delta c + \iota_{\delta\xi}  \xi^\dag,}
\end{equation*}
\begin{align}\label{GR-bulk-action}
S_{GR}=\mathrm{Tr}\intl_{M} & eF_\omega + e^\dag \left( L_\xi^\omega e - [c,e]\right) + \omega^\dag \left(\iota_\xi F_\omega - d_\omega c\right)\nonumber\\ &+\frac12 c^\dag\left(\iota_\xi\iota_\xi F_\omega  - [c,c]\right) + \frac12\iota_{[\xi,\xi]}\xi^\dag, 
\end{align}
and the vector field $Q_{GR}$, satisfying $\iota_{Q_{GR}} \varpi_{GR} = \delta S_{GR}$ when $M$ is closed and without boundary, is given by:
\begin{subequations}\label{Q_GR-bulk1}
\begin{eqnarray}
 &Qe= L_\xi^\omega e - [c,e] & Q\omega= \iota_\xi F_\omega - d_\omega c\\
 &Qc= \frac12\left(\iota_\xi\iota_\xi F_\omega  - [c,c]\right) & Q\xi= \frac12 [\xi,\xi]\\
 &Qe^\dag= F_{\omega} + L_\xi^\omega e^\dag - [c, e^\dag] &Qc^\dag= -[e^\dag,e] - d_\omega \omega^\dag +[c^\dag, c]
 \end{eqnarray}
\begin{align}
&Q\omega^\dag = d_{\omega} e - \iota_\xi[ e^\dag, e]- d_\omega (\iota_\xi \omega^\dag)-[c, \omega^\dag] + \frac12 d_\omega (\iota_\xi \iota_\xi c^\dag)\\
&Q \xi^\dag_\bullet = - e ^\dag_{\bullet} d_\omega e-  d_\omega e^\dag e_\bullet - \omega^\dag_\bullet F_\omega + \iota_\xi c^\dag_\bullet F_\omega+ \partial_\bullet \xi^a \xi^\dagger_a+ \partial_a \xi^a \xi^\dag_\bullet, 
\end{align}
\end{subequations}
where we dropped the GR-subscript and denoted the 1-form coefficient of elements in {$\Omega^{\text{top}}[-2](M,T^*M)$} with a bullet.
\end{definition}
{
\begin{remark}\label{Antighostxi}
The \emph{antighost} field {$\xi^\dag\in\Omega^{\text{top}}[-2](M,T^*M)$} is an element in the fiber of $T^*[-1](\Gamma[1]TM)$. In order to treat it homogeneously with respect to all other fields we can equivalently view it as {$\xi^\dag=\chi\otimes \mathsf{V} \in \Omega^1(M)\otimes \Omega^{\text{top}}(M, \wedge^3 \mathcal{V})$ multiplying it by a fixed normalised volume form $\mathsf{V}\in\Omega^{\text{top}}(M, \wedge^3 \mathcal{V})$}. Its trace will recover the original {$\xi^\dag$,} and it will be particularly useful to simplify the expressions appearing in Proposition \ref{prop:BFVdata}.
\end{remark}
}

We can now state the main results in this paper. {They refer to Definitions \ref{def:extendabletheories} and \ref{def:n-ext-strongBVeq} respectively}. Sections \ref{sec:BVBFVGRproof} and \ref{sec:BF-GR_equivalence} will be devoted to their proof. 
\begin{theorem}\label{thm:BVBFVGR}
The BV theory $\mathfrak{F}^{\uparrow 0}_{GR}=(\mathcal{F}_{GR},S_{GR}, \alpha_{GR}, Q_{GR} )$ is fully extendable.
\end{theorem}
\begin{theorem} \label{thm:extended-equivalence}
The fully extended BV-BFV theories $\mathfrak{F}^{\uparrow 3}_{GR}$ and $\mathfrak{F}^{\uparrow 3}_{\BFnd}$ are strongly equivalent.
\end{theorem}

\begin{remark}
The space of BV fields for classical GR and the action of symmetries presented in Definition \ref{def:BVGR} is essentially independent of spacetime dimensions (although the action functional is not). However, it was shown in \cite{CSPCH} that the 4-dimensional BV theory of General Relativity (in the tetrad formalism) cannot be extended without extra assumptions on the fields. Theorem \ref{thm:BVBFVGR}, compared with the no-go result in \cite{CSPCH} marks a stark difference between 3 and 4 spacetime dimensions. GR in the Einstein--Hilbert formalism, instead, is independent of spacetime dimensions and is always at least 1-extendable {(for $d\not=2$, and whenever the boundary is either space-like or time-like)} \cite{CSEH}.
\end{remark}

\section{Fully extended BV-BFV structure of GR}\label{sec:BVBFVGRproof}
In this section we will prove Theorem \ref{thm:BVBFVGR} by extending the BV theory of definition \ref{def:BVGR} step by step, thus building the fully extended BV-BFV data for GR at every codimension.

First, let us introduce some useful notation and explain the common strategy employed at every step. Throughout the section, as we will only consider GR theory, we will drop the GR subscript everywhere, except in stating the results.

\subsection{Notation and strategy}\label{sec:Notation}
Consider  $X\in \Omega^\bullet(M, \mathcal{V})$ and $Y\in \Omega^\bullet(M, \wedge^2 \mathcal{V})$. Since the image of a nondegenerate triad $e$ is a basis of $V$ at every point, we can express $X$ and $Y$ as 
$$ X = \sum_{i=1}^3 X^{\comp{i}} e_i \qquad \qquad Y =  \sum_{i,j=1, i\neq j}^3 Y^{\comp{ij}} e_i \wedge e_j $$
where $e_i = e (\partial_i)$, and we denote by $X^{\comp{i}}$ the $i$-th component of $X$ and by $Y^{\comp{ij}}$ the $ij$-th component of $Y$. We can then define the following projections, $i=1,2,3$:
\begin{align*}
p_i:   \Omega^\bullet\left(M, \mathcal{V}\right) &  \rightarrow \Omega^\bullet\left(M,  \mathcal{V}\right)\\ X &\mapsto p_i(X)=X^{\comp{i}}e_i.
\end{align*}
It is also useful to define a \textit{dual} map acting on elements of $\Omega^1(M, \wedge^2 \mathcal{V}):$
\begin{align}\label{projection_p_i^dag}
p^\dag_i:   \Omega^\bullet\left(M, {\wedge}^2 \mathcal{V}\right)  &\rightarrow \Omega^\bullet\left(M, {\wedge}^2 \mathcal{V}\right)\\
Y & \mapsto p_i^\dag (Y) =Y^{\comp{hk}} e_h \wedge e_k, \qquad h,k\neq i. \nonumber
\end{align}

{ Furthermore, given any element $Z \in \Omega^\bullet(M, \wedge^3 \mathcal{V})$ we can expand it with respect to the basis spanned by the image of $e$: $Z= Z^{\comp{abc}}e_a\wedge e_b \wedge e_c$. We then define 
\begin{equation}\label{componenttopform}
Z^{\comp{a}} \colon = Z^{\comp{abc}} e_b \wedge e_c \qquad Z^{\comp{a}}\in \Omega^k(M, \wedge^2 \mathcal{V}).
\end{equation}

Let $M$ be a $d$-dimensional smooth manifold. Given a 1-stratification $\{ \filtint{k}\}_{k=0,1}$ of $M$, consider a tubular neighbourhood $U \subset \filtint{0}$ of the embedding $\iota^{\filt{1}}\colon \filtint{1} \to \filtint{0}$ and a local chart on $U$ with coordinates $(x^1, \dots , x^{d-1}, x^n)$ where $x^n$ is the coordinate along the normal direction of the tubular neighbourhood and $ x^1, \dots , x^{d-1}$ are coordinates of $\filtint{1}$. Throughout the article we keep this notation, thus denoting with $n$ the coordinate normal to the first stratum.  We can expand forms $\mu \in \Omega^k(M)$ with respect to this coordinate system and get
\begin{align} \label{restriction}
\mu &=:  \mu^{\paral} + \mu_n dx^n,\\ \notag
\mu^{\paral} &= \sum_{\substack{i_j=1\\ j=1\dots k}}^{d-1} \mu_{i_1\dots i_k} dx^1 \wedge \dots \wedge dx^k\\\notag
\mu_n&=  \sum_{\substack{i_j=1\\  j=1\dots k-1}}^{d-1} \mu_{i_1\dots i_{k-1}n} dx^1 \wedge \dots \wedge dx^{k-1} ,
\end{align}
A similar prescription for the tangent bundle yields:	
\begin{equation}\label{vectorboundary}
\sum_{\mu=1}^3 \xi^\mu \frac{\partial}{\partial x^\mu}= \sum_{a=1}^2\xi^a \frac{\partial}{\partial x^a}+ \xi^n \frac{\partial}{\partial x^n},
\end{equation}
and the contraction $\iota_\xi\mu = \iota_\xi\mu^{\paral} + \mu_n\xi^n + \iota_\xi\mu_n dx^n$ restricts\footnote{Here we obviously mean the pullback of forms along $\iota^{\filt{1}}$.} to $\filtint{1}$ as $\iota_\xi\mu\vert_{\filtint{1}} = \iota_\xi\mu^{\paral} + \mu_n\xi^n$. In particular, note that both $\mu^{\paral}$ and $\mu_n$ restrict to $\filtint{1}$. In case no confusion can arise we will simply denote $\mu^{\paral}$ with $\mu$. Furthermore, we denote with $\mu_b$ the component of $\mu$ in direction of $dx^b$ analogously to the notation $\mu_n$ in \eqref{restriction}. 

\begin{remark}\label{rem:notationhighercodimension}
For higher codimension strata, \eqref{restriction} and \eqref{vectorboundary} are modified accordingly. Though, for the sake of clarity, we use different letters: $m$ (resp. $a$) for the {direction transversal} to the second (resp. third) stratum.
\end{remark}

\begin{remark}\label{rem:coordinatedependingquantities}
The decomposition defined above clearly depends on the choice of the coordinate system compatible with the topological structure of the tubular neighborhood. However, this choice becomes relevant only when one wants to explicitly write a map between spaces of different codimension in a coordinate chart. In such circumstance, as mentioned, this choice will depend on an embedding $\filtint{k}\to\filtint{k+1}$ (Corollaries \ref{cor:BFVdata-projection}, \ref{cor:BFFVdata-projection} and \ref{cor:BFFFVdata-projection}). {The BFV data we will construct for each stratum in Propositions \ref{prop:BFVdata}, \ref{prop:BFFVdata} and \ref{prop:BFFFVdata}, however, is coordinate independent at every codimension.}
\end{remark}

We outline here the strategy { we employ to construct the maps} $\pi^{(k)}$ of Definition \ref{def:n-extended-BVBFV}, as was introduced in \cite{CSEH}. We analyse here the extension from a generic codimension-$k$ stratum  to a codimension-$(k+1)$ stratum. 

The goal is to construct data corresponding to the codimension-$(k+1)$ stratum (the $(k+1)$-extended theory). We will assume here for simplicity that $M^{\filt{k+1}}=\partial M^{\filt{k}}$ is a boundary. The first step is to consider the variation of the action functional in the bulk and to construct appropriate data to satisfy \eqref{rCME} and the axioms in Definition \ref{def:n-extended-BVBFV}. The variation of $\delta S^{\filt{k}}$ will consist of two terms generated by integration by parts: the bulk term will be interpreted as $ \iota_{Q^{\filt{k}}} \varpi^{\filt{k}}$ --- defining the Euler--Lagrange equations for the variational problem --- while the remainder, a boundary term, is interpreted as a one-form $\check{\alpha}^{\filt{k+1}}$ on some appropriate space (see below). Namely, we have an equation formally equivalent to \eqref{rCME}
$$ 
\delta S^{\filt{k}}= \iota_{Q^{\filt{k}}} \varpi^{\filt{k}} - {\check{\pi}^{\filt{k}*}} \check{\alpha}^{\filt{k+1}},
$$ 
with the difference that the correction term ${\check{\pi}^{\filt{k}*}} \check{\alpha}^{\filt{k+1}}$ lives on the intermediate space $\check{\mathcal{F}}^{\filt{k+1}}$, defined as the space of fields and transversal jets restricted to the boundary, with ${\check{\pi}^{\filt{k}*}} \colon \mathcal{F}^{\filt{k}} \to \check{\mathcal{F}}^{\filt{k+1}}$ being the restriction of fields to $M^{\filt{k+1}}$. We will call $\check{\mathcal{F}}^{\filt{k+1}}$ space of \emph{pre-boundary} fields.

We define $\check{\varpi}^{\filt{k+1}} \coloneqq \delta\check{\alpha}^{\filt{k+1}}$, which is a closed form, but in general degenerate and hence not symplectic. However, one expects it to be pre-symplectic i.e. such that the kernel of the associated map 
\begin{align*}
\check{\varpi}^{\filt{k+1}\sharp}: T \check{\mathcal{F}}^{\filt{k+1}}& \rightarrow T^* \check{\mathcal{F}}^{\filt{k+1}}\\
X & \mapsto \check{\varpi}^{\filt{k+1}\sharp}(X)=\check{\varpi}^{\filt{k+1}}(X, \cdot)
\end{align*}
is \textit{regular}. This condition, which has to be explicitly checked, shows that the kernel is a subbundle in $T\check{\mathcal{F}}^{\filt{k+1}}$, and one can define the space of boundary fields to be the symplectic reduction\footnote{{Typically, the procedure we employ will provide a chart for the quotient, showing that the symplectic reduction is indeed smooth.}}:
\begin{equation}\label{symplectic-reduction}
	\mathcal{F}^{\filt{k+1}}\coloneqq \qsp{\check{\mathcal{F}}^{\filt{k+1}}}{\mathrm{ker}(\check{\varpi}^{\filt{k+1}\sharp})}
\end{equation}

The symplectic reduction map $\pi_{\mathrm{red}}^{\filt{k+1}}\colon \check{\mathcal{F}}^{\filt{k+1}}\longrightarrow \mathcal{F}^{\filt{k+1}}$ can be computed in a chart by explicitly flowing along the vertical vector fields (i.e. $X \in T \check{\mathcal{F}}^{\filt{k+1}}$ such that $\check{\varpi}^{\filt{k+1}\sharp}(X)=0$). The  projection $\pi^{\filt{k}}: \mathcal{F}^{\filt{k}} \rightarrow \mathcal{F}^{\filt{k+1}}$ to the true space of boundary fields, is then obtained by composing $\check{\pi}^{\filt{k}}$ with the symplectic reduction map\footnote{
The numbering convention we are using is such that the codimension index of the various maps coincides with that of their domain.} $\pi_{\mathrm{red}}^{\filt{k+1}}$, i.e. $\pi^{\filt{k}}\coloneqq \pi_{\mathrm{red}}^{\filt{k+1}}\circ\check{\pi}^{\filt{k}}$. { To summarise, we have
\begin{equation}
    \xymatrix{
        \ar@/_2pc/[dd]_{{\pi}^{\filt{k}}} \mathcal{F}^{\filt{k}} \ar[d]^{\check{\pi}^{\filt{k}}} &   \varpi^{\filt{k}},\ (k-1)\text{-symplectic}\\
        \check{\mathcal{F}}^{\filt{k+1}} \ar[d]^{\pi_{\mathrm{red}}^{\filt{k}}} & \check\varpi^{\filt{k+1}},\ k\text{-presymplectic}\\
        \mathcal{F}^{\filt{k+1}} & \varpi^{\filt{k+1}},\ k\text{-symplectic}
    }
\end{equation}
See the proof of Proposition \ref{prop:BFVdata} for more details.}

To recover the rest of the BV-BFV data on the higher codimension stratum, we first define a \textit{pre-boundary} action functional $\check{S}^{\filt{k+1}}$ on $\check{\mathcal{F}}^{\filt{k+1}}$ \emph{via} an analogue of Equation \eqref{b-action}: 
$$
\iota_{Q^{\filt{k}}} \iota_{Q^{\filt{k}}} \varpi^{\filt{k}} = 2 \check{\pi}^{\filt{k} *} \check{S}^{\filt{k+1}}.
$$
It then follows from the BV-BFV theorems {\cite[Section 3.1.3]{CMR2012}} that $\check{S}^{\filt{k+1}}$ is basic with respect to $\pi^{\filt{k+1}}_{\text{red}}: \check{\mathcal{F}}^{\filt{k+1}} \rightarrow \mathcal{F}^{\filt{k+1}}$: namely, there exists $S^{\filt{k+1}}$ on $\mathcal{F}^{\filt{k+1}}$ such that $\pi^{\filt{k+1}*}_{\text{red}} S^{\filt{k+1}}= \check{S}^{\filt{k+1}}$, and consequently 
 $$
\iota_{Q^{\filt{k}}} \iota_{Q^{\filt{k}}} \varpi^{\filt{k}} = 2 {\pi}^{\filt{k}*}{S}^{\filt{k+1}}.
$$ 
 
{
Moreover, \cite[Section 3.1.3]{CMR} establishes that on $\mathcal{F}^{\filt{k+1}}$ one has the evolutionary, cohomological, projectable vector field $Q^{\filt{k+1}}=\pi^{\filt{k}}_*Q^{\filt{k}}$, which will be Hamiltonian w.r.t. $S^{\filt{k+1}}$ up to higher codimension terms.
}
}

\subsection{1-extended GR theory }
If we allow $M$ to bear a $1$-stratification, and denote codimension-$1$ strata by $M^{\filt{1}}$, denoting by $\Omega_{nd}^1(\filtint{1}, \mathcal{V})$ the space of maps $e\colon T M^{\filt{1}} \longrightarrow \mathcal{V}$ such that the image of $e$ is a linearly independent system, we have the following.

\begin{proposition}\label{prop:BFVdata}
The BV theory $\mathfrak{F}^{\uparrow 0}_{GR}=(\mathcal{F}_{GR}, S_{GR}, \varpi_{GR}, Q_{GR})$ is 1-extendable  to $\mathfrak{F}^{\uparrow 1}_{GR}$. The codimension-$1$ data are: 
\begin{itemize}
\item The space of codimension-$1$ fields, given by the bundle
\begin{equation}
\mathcal{F}^{\filt{1}}_{GR} \longrightarrow \Omega_{nd}^1(\filtint{1}, \mathcal{V}),
\end{equation}
with local trivialisation on an open $\mathcal{U}^{\filt{1}} \subset \Omega_{nd}^1(\filtint{1}, \mathcal{V})$
\begin{equation}\label{LoctrivF1}
\mathcal{F}^{\filt{1}}_{GR}\simeq \mathcal{U}^{\filt{1}} \times \Omega^1( \filtint{1}, \wedge^2\mathcal{V})\oplus T^* \left(\Omega^0[1]( \filtint{1}, \wedge^2\mathcal{V})\oplus \mathfrak{X}[1](\filtint{1}) \oplus C^\infty[1](\filtint{1})\right),
\end{equation}
and fields denoted by $\te\in \mathcal{U}^{\filt{1}}$ and $\tom\in\Omega^1(\filtint{1},\wedge^2\mathcal{V})$ in degree zero, $\tc\in\Omega^0[1](\filtint{1},\wedge^2\mathcal{V})$, $\txi{}\in\mathfrak{X}[1](\filtint{1})$ and $\txi{n}\in C^\infty[1](\filtint{1})$ in degree one, $\tom^\dag\in\Omega^3[-1](\filtint{1},\mathcal{V})$ and $\te^\dag\in\Omega^{2}[-1](\filtint{1},\wedge^2 \mathcal{V})$ in degree minus one, together with a fixed vector field $\epsilon_n \in \Gamma(\mathcal{V})$, completing the image of elements $\te\in\mathcal{U}^{\filt{1}}$ to a basis of  $\mathcal{V}$;
\item The codimension-$1$ one-form, symplectic form and action functional
\begin{subequations}\label{GR_BFV_data}\begin{align}
\alpha^{\filt{1}}_{GR}&= \trintl{\filtint{1}} - \te \delta \tom+ \tom^\dag \delta \tc -\te^{\dag} \epsilon_n \delta \txi{n}- (\iota_{\delta \txi{}} \te) \te^{\dag}+ \iota_{\txi{}}\tom^\dag \delta \tom,\\\label{BoundarytwoformGR}
\varpi^{\filt{1}}_{GR}&= \trintl{\filtint{1}} - \delta\te \delta \tom+  \delta\tom^\dag \delta \tc - \delta\te^{\dag} \epsilon_n \delta \txi{n}+ \iota_{\delta \txi{}} \delta({\te_\bullet\otimes} \te^{\dag})+ \delta(\iota_{\txi{}}\tom^\dag )\delta \tom,\\\label{BoundaryactionGR}
S^{\filt{1}}_{GR}&= \trintl{\filtint{1}} - \iota_{\txi{}} \te F_{\tom} -\epsilon_n \txi{n} F_{\tom} - \tc d_{\tom} \te + \frac12 [\tc,\tc]\tom^{\dag} +\frac12 \iota_{\txi{}} \iota_{\txi{}} F_{\tom}\tom^\dag+ \frac12 \iota_{[\txi{},\txi{}]}\te\te^\dag \nonumber\\ &+\tc d_{\tom}( \iota_{\txi{}} \tom^\dag)+ L_{\txi{}}^{\tom} (\epsilon_n \txi{n}) \te^\dag- [\tc, \epsilon_n \txi{n}]\te^\dag,
\end{align}\end{subequations}
where $\iota_{\delta \txi{}} \delta({\te_\bullet\otimes} \te^{\dag})$ indicates that the contraction is considered with respect to the one-form\footnote{{In a local chart we can write $\iota_{\delta \txi{}} \delta({\te_\bullet\otimes} \te^{\dag}) = \delta\xi^a \delta(e_a e^\dag)$.}} $\te$;
\item The cohomological vector field $Q^{\filt{1}}_{GR}$
\begin{subequations}\label{Q_GR-boundary1}
\begin{align}
 Q^{\filt{1}}_{GR} {\te} &=- L_{\txi{}}^{\tom} \te + d_{\tom}(\epsilon_n \txi{n})+[\tc, \te]- X_n^{\comp{a}} \tom_a^\dag+Y_n^{\comp{a}} \tom_a^\dag  \\
 Q^{\filt{1}}_{GR} {\tom} &= - \iota_{\txi{}} F_{\tom} + d_{\tom}\tc - X_n^{\comp{a}} \te_a^\dag+Y_n^{\comp{a}} \te_a^\dag\\
 Q^{\filt{1}}_{GR} {\tc} &= \frac{1}{2}[\tc, \tc]- \frac{1}{2} \iota_{\txi{}} \iota_{\txi{}} F_{\tom}- \iota_{\txi{}}X_n^{\comp{a}} \te_a^\dag+\iota_{\txi{}}Y_n^{\comp{a}} \te_a^\dag\\
 Q^{\filt{1}}_{GR} {\tom^\dag} &=- d_{\tom}\te + [\tc, \tom^\dag] + d_{\tom}\iota_{\txi{}} \tom ^\dag - [\epsilon_n \txi{n},\te^\dag]\\
 Q^{\filt{1}}_{GR} {\te^\dag} &=-F_{\tom}+ [\tc, \te^\dag] + d_{\tom}\iota_{\txi{}} \te ^\dag\\
 Q^{\filt{1}}_{GR} {\txi{}} &=- \frac{1}{2}[\txi{},\txi{}] - \left(X_n^{\comp{a}}+Y_n^{\comp{a}}\right)\frac{\partial}{\partial x^a}  \\
 Q^{\filt{1}}_{GR} {\txi{n}} &=+ X_n^{\comp{n}} -Y_n^{\comp{n}}
\end{align}
\end{subequations}
where $X_n= L_{\txi{}}^{\tom}(\epsilon_n \txi{n})$ and $Y_n= [\tc, \epsilon_n \txi{n}]$, and the superscript ${}^{\comp{n}}$ denotes the component with respect to $\epsilon_n$, following the notation of Section \ref{sec:Notation};
\item The projection to codimension-$1$ fields ${\pi}^{\filt{0}}_{GR} = {\pi_{\mathrm{red}}}^{\filt{1}} \circ {\check{\pi}}^{\filt{0}}$ is smooth, where ${\check{\pi}}^{\filt{0}}: \mathcal{F}^{\filt{0}}\to \check{\mathcal{F}}^{\filt{1}}$ is the restriction of the codimension-$0$ fields to the $1$-stratum $M^{\filt{1}}$ and ${\pi_{\mathrm{red}}}^{\filt{1}}: \check{\mathcal{F}}^{\filt{1}}\to \mathcal{F}^{\filt{1}}$ is the symplectic reduction \eqref{symplectic-reduction}.
\end{itemize}
\end{proposition}

\begin{remark}
Before we produce the proof of Proposition \ref{prop:BFVdata}, let us observe that all the expressions in \eqref{Q_GR-boundary1} (and more manifestly in \eqref{GR_BFV_data}) are coordinate independent. Indeed,  expressions like $\iota_{\txi{}}X_n^{\comp{a}}\te_a^\dag$ or $X^{\comp{a}}\frac{\partial}{\partial x^a}$ are tensorial in a local chart of $M^{(1)}$, while the component ${}^{\comp{n}}$ refers to the completion of a basis $\{\mathrm{Im}(\te), \epsilon_n\}$. Furthermore, the projection $\pi^{\filt{0}}_{GR}$ as defined above is coordinate independent (both restrictions and symplectic reductions can be defined without the use of a coordinate system). However, the explicit expression of the BFV map $\pi^{\filt{0}}_{GR}$ presented later, in Corollary \ref{cor:BFVdata-projection}, depends on the choice of an embedding $\filtint{1}$ in $\filtint{0}$ and a chart on a tubular neighbourhood (cf. Remark \ref{rem:coordinatedependingquantities}).
\end{remark}

\begin{proof}
From the variation of the action \eqref{GR-bulk-action}, following the strategy outlined in subsection \ref{sec:Notation}, we get the pre-boundary one-form by isolating the boundary terms:
\begin{align*}
\check{\alpha}^{\filt{1}}=\trintl{\filtint{1}} & -e\delta\omega + e^\dag_n\xi^n\delta e + e^\dag\delta(e_n\xi^n) + e^\dag\iota_{\delta\xi}e  + \iota_\xi\omega^\dag\delta\omega + \omega^\dag_n\xi^n\delta\omega + \omega^\dag\delta c \\
-& \left( \iota_\xi c^\dag_n\xi^n\right)\delta\omega - \xi^n\iota_{\delta\xi}\chi\mathsf{V} -  \xi^n\delta\xi^n\chi_n\mathsf{V}
\end{align*}
where $\xi^{\dag}= \chi\otimes \mathsf{V} \equiv \chi \mathsf{V}$. Then, we derive the two form $\check{\varpi}^{\filt{1}}= \delta \check{\alpha}^{\filt{1}}$:
\begin{align*}
\check{\varpi}^{\filt{1}}=\trintl{\filtint{1}}& -\delta e\delta\omega + \delta e^\dag_n\xi^n\delta e + e^\dag_n\delta\xi^n\delta e + \delta e^\dag\delta(e_n\xi^n) + \delta e^\dag\iota_{\delta\xi} e + e^\dag \iota_{\delta\xi}\delta e \\ &+  \iota_{\delta\xi}\omega^\dag \delta\omega
 + \delta(\omega^\dag_n\xi^n)\delta\omega +  \iota_\xi\delta\omega^\dag\delta\omega + \delta\omega^\dag\delta c - \delta\left( \iota_\xi c^\dag_n\xi^n\right)\delta\omega \\
 &- \delta\xi^n\iota_{\delta\xi}\chi\mathsf{V} - \xi^n\iota_{\delta\xi}\delta\chi \mathsf{V} -\delta\xi^n\delta\xi^n\chi_n\mathsf{V} + \xi^n\delta\xi^n\delta\chi_n\mathsf{V}.
\end{align*}
We want to make sure that $\check{\varpi}^{\filt{1}}$ is pre-symplectic: the kernel of such a two-form is defined by the equations
\begin{subequations}\label{kernelsettom1}\begin{align}
\X{e} &= \iota_{\X{\xi}}\omega^\dag - \X{\xi^n}\omega^\dag_n+ \iota_\xi\X{\omega^\dag} + \X{\omega^\dag_n}\xi^n  - \X{\iota_\xi c_n^\dag \xi^n}\\
\X{\omega} &= \iota_{\X{\xi}}e^\dag + \X{\xi^n}e^\dag_n + \X{e^\dag_n}\xi^n\\
\X{\xi^\mu}  e_{\mu}&= -\X{e_n}\xi^n \label{xi-e} \\
\X{c}&=\iota_\xi\X{\omega}\\
\X{\omega^\dag}&=0\\
e_n(X_{e^\dag})&= -(X_e)e^\dag_n + (X_\omega)\omega^\dag_n  - (X_\omega) c^\dag_{nb}\xi^b \nonumber\\ \label{edagn} 
&+ \left(2(X_{\xi^n})\chi_n + (X_{\chi_n})\xi^n +(X_{\xi^a})\chi_a\right)\mathsf{V}\\\label{edaga}
e_a(X_{e^\dag})&=-(X_e)e^\dag_a + (X_\omega)\omega_a^\dag - (X_\omega)c^\dag_{an}\xi^n  +\left( (X_{\xi^n})\chi_a  + (X_{\chi_a})\xi^n\right)\mathsf{V},
\end{align}\end{subequations}
together with
\begin{subequations}\label{kernelset2}
\begin{align}
\X{e^\dag}\xi^n&=0\\
\iota_\xi\X{\omega}\xi^n &=0\\
\X{e}\xi^n &=0\\
\X{\omega}\xi^n&=0\\
\X{\xi}^\rho\xi^n&=0.
\end{align}
\end{subequations}
We first solve \eqref{xi-e}: expanding $ \X{e_n}$ in the basis $\{e_{\mu}\}_{\mu =1,2,n}$ we get a vector equation whose single components are
$$\X{\xi^\mu}  = -\X{e_n}^{\comp{\mu}}\xi^n \text{  for } \mu= 1,2,n.$$
It is then possible to solve equations \eqref{kernelsettom1} to yield
\begin{subequations}\label{kernelset3}\begin{align}
\X{\omega^\dag}&=0\\
\X{\xi^\mu}  &= -\X{e_n}^{\comp{\mu}}\xi^n\\
\X{\omega}&=-\X{e_n}^{\comp{b}}\xi^n e^\dag_b - \X{e_n}^{\comp{n}}\xi^n e^\dag_n + \X{e_n^\dag}\xi^n\\\notag
\X{e}&=\X{e_n}^{\comp{b}}\xi^n\omega^\dag_b + \X{e_n}^{\comp{n}}\xi^n\omega^\dag_n + \X{\omega^\dag_n}\xi^n 
-\iota_\xi\X{c^\dag_n}\xi^n\\ &+ \X{e_n}^{\comp{n}}\iota_\xi c^\dag_n\xi^n \\
\X{c}&= \iota_\xi \X{\omega}
\end{align}
\begin{align}\notag
p_a^{\dag}\X{e^\dag}&= \Big[\big(\X{e_n}^{\comp{b}}\xi^n\omega^\dag_b + \X{e_n}^{\comp{n}}\xi^n\omega^\dag_n + \X{\omega^\dag_n}\xi^n 
-\iota_\xi\X{c^\dag_n}\xi^n\big)e^\dag_a \\\notag &+ \X{e_n}^{\comp{n}}\iota_\xi c^\dag_n\xi^n e^\dag_a+ \big(-\X{e_n}^{\comp{b}}\xi^n e^\dag_b - \X{e_n}^{\comp{n}}\xi^n e^\dag_n + \X{e_n^\dag}\xi^n\big)\omega_a^\dag \\\notag &- \big(-\X{e_n}^{\comp{b}}\xi^n e^\dag_b - \X{e_n}^{\comp{n}}\xi^n e^\dag_n + \X{e_n^\dag}\xi^n\big)c^\dag_{an}\xi^n\Big]^{\comp{a}}\\ & - \X{e_n}^{\comp{n}}\xi^n\chi_a \mathsf{V}^{\comp{a}} + \X{\chi_\bullet}\xi^n\mathsf{V}^{\comp{a}}
\end{align}
\begin{align}\notag
p_n^{\dag}\X{e^\dag}&= \Big[ \big(\X{e_n}^{\comp{b}}\xi^n\omega^\dag_b + \X{e_n}^{\comp{n}}\xi^n\omega^\dag_n + \X{\omega^\dag_n}\xi^n 
-\iota_\xi\X{c^\dag_n}\xi^n\big)e^\dag_n\\\notag
&+ \X{e_n}^{\comp{n}}\iota_\xi c^\dag_n\xi^n e^\dag_n+ \big(-\X{e_n}^{\comp{b}}\xi^n e^\dag_b - \X{e_n}^{\comp{n}}\xi^n e^\dag_n + \X{e_n^\dag}\xi^n\big)\omega_n^\dag \\\notag
&- \big(-\X{e_n}^{\comp{b}}\xi^n e^\dag_b - \X{e_n}^{\comp{n}}\xi^n e^\dag_n + \X{e_n^\dag}\xi^n\big)c^\dag_{nb}\xi^b\\
&- 2\X{e_n}^{\comp{n}}\xi^n\chi_n\mathsf{V} + \X{\chi_n}\xi^n\mathsf{V} -\X{e_n}^{\comp{a}}\xi^n\chi_a\mathsf{V}\Big]^{\comp{n}}
\end{align}\end{subequations}
and, since Equations \eqref{kernelset3} force the left hand sides to be proportional to the odd field $\xi^n$, Equations \eqref{kernelset2} are automatically satisfied. This shows that the kernel has \emph{constant rank} --- i.e. $\check{\varpi}^{\filt{1}}$ is pre-symplectic --- and we can perform symplectic reduction.

We now compute the  BV-BFV data by presenting a chart for the symplectic reduction 
$$
\mathcal{F}_{GR}^{\filt{1}}\coloneqq \qsp{\check{\mathcal{F}}^{\filt{1}}_{GR}}{\mathrm{ker}(\check{\varpi}^{\filt{1}\sharp})}.
$$
We flow along vertical vector fields (i.e. vector fields in the kernel of $\check{\varpi}^{\filt{1}\sharp}$) to obtain boundary coordinates. In other words, denoting by $\varphi_{Y}$ the flow of a vector field $Y$ at time $s=1$, we define the change of coordinates on a field $\psi$ to be given by 
$$
\widetilde{\psi} \coloneqq (\varphi_{\mathbb{E}_n}\circ\varphi_{\mathbb{E}^\dag_n}\circ\varphi_{\mathbb{C}^\dag_n}\circ \varphi_{\mathbb{\Omega}^\dag_n}\circ \varphi_{\mathbb{X}_n}\circ\varphi_{\mathbb{X}_a})(\psi),
$$
where the vertical vector fields $\mathbb{E}_n,\mathbb{E}_n^\dag,\mathbb{C}_n,\mathbb{\Omega}_n^\dag,\mathbb{X}_n$ and $\mathbb{X}_a$ read
\begin{subequations}\begin{align}\notag
\mathbb{E}_n &= \X{e_n}\pard{}{e_n} - \left(\X{e_n}^{\comp{a}} e^\dag_a \xi^n+ \X{e_n}^{\comp{n}} e^\dag_n\xi^n\right)\pard{}{\omega}\\\notag 
&- \iota_\xi\left(\X{e_n}^{\comp{a}} e^\dag_a \xi^n+ \X{e_n} e^\dag_n\xi^n\right)\pard{}{c}\\\notag
 &+ \left(-\X{e_n}^{\comp{a}} \omega^\dag_a\xi^n -\X{e_n}^{\comp{n}}\omega^\dag_n\xi^n + \X{e_n}^{\comp{n}}\iota_\xi c^\dag_n\xi^n\right)\pard{}{e} - \X{e_n}^\rho\xi^n\pard{}{\xi^\rho}\\\notag
 &-\Big[\big(\X{e_n}^{\comp{b}}\xi^n\omega^\dag_b + \X{e_n}^{\comp{n}}\xi^n\omega^\dag_n + \X{e_n}^{\comp{n}}\iota_\xi c^\dag_n\xi^n\big)e^\dag_a \\\notag &+ \big(\X{e_n}^{\comp{b}}\xi^n e^\dag_b + \X{e_n}^{\comp{n}}\xi^n e^\dag_n \big)\omega_a^\dag+ \big(\X{e_n}^{\comp{b}}\xi^n e^\dag_b + \X{e_n}^{\comp{n}}\xi^n e^\dag_n \big)c^\dag_{an}\xi^n\\ \notag& - \X{e_n}^{\comp{n}}\xi^n\chi_a \mathsf{V}\Big]^{\comp{a}}  \pard{}{p_a^{\dag}e^\dag}\\\notag
&+\Big[ \big(\X{e_n}^{\comp{b}}\xi^n\omega^\dag_b + \X{e_n}^{\comp{n}}\xi^n\omega^\dag_n + \X{e_n}^{\comp{n}}\iota_\xi c^\dag_n\xi^n\big)e^\dag_n\\\notag
&- \big(\X{e_n}^{\comp{b}}\xi^n e^\dag_b + \X{e_n}^{\comp{n}}\xi^n e^\dag_n \big)\omega_n^\dag + \big(-\X{e_n}^{\comp{b}}\xi^n e^\dag_b + \X{e_n}^{\comp{n}}\xi^n e^\dag_n \big)c^\dag_{nb}\xi^b\\
&- 2\X{e_n}^{\comp{n}}\xi^n\chi_n\mathsf{V}  -\X{e_n}^{\comp{a}}\xi^n\chi_a\mathsf{V}\Big]^{\comp{n}}
\pard{}{p_n^{\dag}e^\dag}
\end{align}
\begin{align}\notag
\mathbb{E}^\dag_n &= \X{e^\dag_n}\pard{}{e^\dag_n} + \X{e^\dag_n}\xi^n \pard{}{\omega} + \iota_\xi\X{e^\dag_n}\xi^n \pard{}{c} \\
&+ \left[\X{e^\dag_n}\xi^n\omega^\dag_a\right]^{\comp{a}}\pard{}{p_a^{\dag}e^\dag} + \left[ \X{e^\dag_n}\xi^n\iota_\xi c^\dag_n + \X{e^\dag_n}\xi^n\omega^\dag_n\right]^{\comp{n}}\pard{}{p_n^{\dag}e^\dag}\\
\mathbb{C}^\dag_n&=\X{c^\dag_n}\pard{}{c^\dag_n} - \iota_\xi\X{c^\dag_n}\xi^n\pard{}{e} + \left(\iota_\xi\X{c^\dag_n}\xi^n e^\dag_n \right)^{\comp{n}}\pard{}{p_n^{\dag}e^\dag}\notag \\ &+ \left(\iota_\xi\X{c^\dag_{na}}\xi^n e^\dag \right)^{\comp{a}}\pard{}{p_a^{\dag}e^\dag}\\
\mathbb{\Omega}^\dag_n &= \X{\omega^\dag_n} \pard{}{\omega^\dag_n} + \X{\omega^\dag_n}\xi^n \pard{}{e} - \left(\X{\omega^\dag_n}\xi^n e^\dag_n \right)^{\comp{n}}\pard{}{p_n^{\dag}e^\dag}- \left(\X{\omega^\dag_{na}}\xi^n e^\dag \right)^{\comp{a}}\pard{}{p_a^{\dag}e^\dag}\\
\mathbb{X}_n&= \X{\chi_n}\pard{}{\chi_n} + \left(\X{\chi_n}\xi^n\mathsf{V}\right)^{\comp{n}}\pard{}{p_n^{\dag}e^\dag}\\
\mathbb{X}_a&=\X{\chi_a}\pard{}{\chi_a} + \X{\chi_a}\xi^n\mathsf{V}^{\comp{a}}\pard{}{p_a^{\dag}e^\dag} 
\end{align}
\end{subequations}

{Using $\mathbb{X}_a$ and $\mathbb{X}_n$ we can eliminate $\chi_a$ and $\chi_n$. From the differential equation $\dot{\chi}_\rho(s)=\X{\chi_\rho}$ ($\rho= a, n$), we conclude that 
$$\chi_\rho(s=1)=0 \iff \X{\chi_\rho}=-\chi_\rho(0),$$ 
with $s$ parametrising the flow of the vector field. The equation induced by the flow of $\mathbb{X}_n$ reads
\begin{equation}
    \frac{d}{ds}(p_n^\dag e^\dag) = \left(-\chi_\rho(0)\xi^n\mathsf{V}\right)^{\comp{n}},
\end{equation}
and similarly for $\mathbb{X}_a$. It follows that, if we define $p_n^{\dag}e^{\dag }[1]\coloneqq (p_n^{\dag}e^{\dag })(s=1)$, we have
\begin{eqnarray}
p_n^{\dag}e^{\dag }[1] =p_n^{\dag} e^\dag - \chi_n\xi^n\mathsf{V}^{\comp{n}}; & p_a^{\dag}e^{\dag}[1] = p_a^{\dag} e^\dag - \chi_a\xi^n\mathsf{V}^{\comp{a}}.
\end{eqnarray} 
The numbered square bracket ${[1]}$ denotes the step-by-step reconstruction of the change of variables: indeed, the variable $p_n^{\dag} e^{\dag}$ gets transformed first by (the flow of) $\mathbb{X}_n$, and then by $\mathbb{\Omega}^\dag_n$ and $\mathbb{E}^\dag_n$. Flowing along each of these vertical vector fields defines a temporary change of coordinates, which will be denoted by ${[k]}$. 

Indeed, using $\mathbb{\Omega}^\dag_n$ we can dynamically set $\omega_n^\dag(s=1)=0$ by picking $\X{\omega^\dag_n}=-\omega^\dag_n(0)$. This choice induces differential equations for the newly defined variables:
\begin{align}
    \frac{d}{ds}(p_n^\dag e^\dag[1]) = (\omega^\dag_n\xi^n e^\dag_n)^{\comp{a}} \iff (p_n^\dag e^\dag[1])(s) = p_n^\dag e^\dag[1](0) + (\omega^\dag_n\xi^n e^\dag_n)^{\comp{n}}s,\\
    \frac{d}{ds}(p_a^\dag e^\dag[1]) = (\omega^\dag_n\xi^n e^\dag_n)^{\comp{a}} \iff (p_a^\dag e^\dag[1])(s) = p_a^\dag e^\dag[1](0) + (\omega^\dag_n\xi^n e^\dag_n)^{\comp{a}}s,
\end{align}
as well as
\begin{equation}
    \frac{d}{ds}e = -\omega^\dag_n \xi^n \iff e(s) = e(0) -\omega^\dag_n \xi^n,
\end{equation}
and defining the new temporary variables by $p_n^\dag e^\dag[2]\coloneqq p_n^\dag e^\dag[1](s=1)$ (and similarly for $p_a^\dag e^\dag[2]$, we obtain:
\begin{eqnarray}
 e[1]\coloneqq e(s=1) = e - \omega^\dag_n\xi^n; \notag\\ (p_n^{\dag}e^\dag){[2]}=(p_n^{\dag}e^{\dag}){[1]} + \left(\omega_n^\dag\xi^n e^\dag_n \right)^{\comp{n}}; & (p_a^{\dag}e^\dag){[2]}=(p_a^{\dag}e^\dag){[1]} + \left(\omega_{na}^\dag\xi^n e^\dag \right)^{\comp{a}}
\end{eqnarray}

Moving on to $\mathbb{E}^\dag_n$ with an analogous procedure, we solve the associated differential equations to yield
\begin{eqnarray}
\omega{[1]}=\omega -e^\dag_n\xi^n; & c{[1]}=c - \iota_\xi e^\dag_n\xi^n; \notag\\
(p_a^{\dag}e^\dag){[3]} = (p_a^{\dag}e^\dag){[2]} - (e^\dag_n\xi^n\omega^\dag_a)^{\comp{a}}; &  (p_n^{\dag}e^\dag){[3]} = (p_n^{\dag}e^\dag){[2]} - (e^\dag_n \iota_\xi c^\dag_n\xi^n)^{\comp{n}}
\end{eqnarray}	
while using $\mathbb{C}^\dag_n$ in the same fashion we can conclude:
\begin{eqnarray}
e{[2]} = e{[1]} +\iota_\xi c^\dag_n \xi^n; &  (p_a^{\dag}e^\dag){[4]} = (p_a^{\dag}e^\dag){[3]} - \left(\iota_\xi c^\dag_{na}\xi^n e^\dag\right)^{\comp{a}}.
\end{eqnarray}}
Notice that we did not consider the coefficient of $\pard{}{e^\dag}$ in $\mathbb{E}^\dag_n$, for have fixed the values $c^\dag_n=\omega_n^\dag=0$ at the internal parameter $s=1$ along the flow of the previously employed vector fields.\\
Now it is time to turn to $\mathbb{E}_n$. Its simplified expression after flowing along the other vector fields is
\begin{align}\notag
\mathbb{E}_n =& \X{e_n}\pard{}{e_n} - \X{e_n}^{\comp{a}} e^\dag_a \xi^n\pard{}{\omega} - \iota_\xi \X{e_n}^{\comp{a}} e^\dag_a \xi^n\pard{}{c}\\ -& \X{e_n}^{\comp{a}} \omega^\dag_a\xi^n \pard{}{e} - \X{e_n}^\rho\xi^n\pard{}{\xi^\rho}
\end{align}
We want to flow along its integrating diffeomorphism and set the field $e_n$ to a given value. In this case we cannot set $e_n(s=1)=0$ because this would violate the nondegeneracy requirement for the triad field. We will fix $e_n(1)$ to a vector $\epsilon_n\in V$ proportional to the original $e_n$, pointwise, and thus linearly independent from (the vectors in the image of) $e$. Observe that for an open subset $\mathcal{U}^{\filt{1}}\subset \Omega_{nd}^1(\filtint{1}, \mathcal{V})$ the choice of $\epsilon_n$ is independent of $e\in \mathcal{U}^{\filt{1}}$. The differential equation $\dot{e}_n = \X{e_n}$ is solved as $e_n(s)=e_n(0) + \X{e_n}s$, so that, fixing $e_n(s=1)\equiv (1+ \varepsilon) e_n(0)$ yields the flow:
\begin{eqnarray}
\X{e_n}=\varepsilon e_n(0); & e_n(s)=e_n(0)(1 + \varepsilon s),
\end{eqnarray}
{ with $\varepsilon \in C^\infty(M^{\filt{1}})$, $\varepsilon >0$}. In order to compute the other flows, we have to consider the components of $ \X{e_n}$ in the (varying) basis vectors $\{e_a\}$. With our choice we have  $\X{e_n} \propto e_n(0) \propto e_n(s)$. Hence we obtain 
$$ \X{e_n}^{\comp{a}} =0 \text{ and } \X{e_n}^{\comp{n}}(s)= \frac{\varepsilon}{1+ \varepsilon s}$$
Hence, looking at the expression for $\mathbb{E}_n$  the only equation that we have to consider is
\begin{equation}
\dot{\xi}^n = - \X{e_n}^{\comp{n}} \xi^n
\end{equation}
and we easily find
$${\xi}^n(s)= \frac{1}{1+ \varepsilon s}{\xi}^n(0)$$
so that $${\xi}^{n}[1]= \frac{1}{1+ \varepsilon }{\xi}^n.$$
Gathering what we have done so far, defining $\epsilon_n:=(1+ \varepsilon) e_n$ (now a fixed vector field), {we have constructed a chart on the symplectic reduction 
$$\mathcal{F}^{\filt{1}}\coloneqq \qsp{\check{\mathcal{F}}^{\filt{1}}}{\mathrm{ker}(\check{\varpi}^{\filt{1}\sharp})}$$ 
so that the map $\pi_{\mathrm{red}}^{\filt{1}}\colon \check{\mathcal{F}}^{\filt{1}}\longrightarrow \mathcal{F}^{\filt{1}}$ reads:}
\begin{equation}\label{inproofProjectionbulktoboundary}
{\pi}^{\filt{1}}_{\text{red}}:
\begin{cases}
\te:=  e - \omega^\dag_n\xi^n +\iota_\xi c^\dag_n\xi^n\\
\tom:= \omega - e^\dag_n\xi^n\\
\tc:=  c -  \iota_\xi e^\dag_n\xi^n \\
\txi{n} := (1+\varepsilon)^{-1}\xi^n\\
\txi{a} := \xi^a \\
\tilde{f_a}^{\dag} :=p_a^{\dag}e^{\dag} - \chi_a\xi^n\mathsf{V}^{\comp{a}}  +\left(\omega_{na}^\dag\xi^n e^\dag \right)^{\comp{a}}- (e^\dag_n\xi^n\omega^\dag_a)^{\comp{a}} - \left(\iota_\xi c^\dag_{na}\xi^n e^\dag\right)^{\comp{a}}\\ 
\tilde{f_n}^{\dag} :=p_n^{\dag}e^{\dag} - \chi_n\xi^n\mathsf{V}^{\comp{n}}  + \left(\omega_n^\dag\xi^n e^\dag_n \right)^{\comp{n}}- (e^\dag_n \iota_\xi c^\dag_n\xi^n)^{\comp{n}}\\
\tom^\dag:=\omega^\dag
\end{cases}
\end{equation}
with $\te\in\mathcal{U}\subset \Omega_{nd}^1(\filtint{1}, \mathcal{V})$, and the BV-BFV map\footnote{Observe that, from now on, we will omit mentioning precomposition with $\check{\pi}$, when no confusion can arise.} $\pi^{\filt{0}} \colon \mathcal{F}^{\filt{0}} \longrightarrow \mathcal{F}^{\filt{1}}$ is 
$$
\pi^{\filt{0}}\coloneqq \pi_{\mathrm{red}}^{\filt{1}}\circ\check{\pi}^{\filt{0}}.
$$ 
This data defines a chart for the (locally trivialised) bundle
$$
\mathcal{F}^{\filt{1}}_{GR}\simeq \mathcal{U}^{\filt{1}} \times \Omega^1( \filtint{1}, \wedge^2\mathcal{V})\oplus T^* \left(\Omega^0[1]( \filtint{1}, \wedge^2\mathcal{V})\oplus \mathfrak{X}[1](\filtint{1}) \oplus C^\infty[1](\filtint{1})\right)
$$
since, for all $\te\in \mathcal{U}$ we can fix a completion $\epsilon_n$ that does not depend on $\te$.

An easy computation shows that $\check{\alpha}^{\filt{1}}$ is not basic (in particular it is not horizontal, i.e. $\iota_{\mathbb{E}_n}\check{\alpha}^{\filt{1}} \neq 0$), but it descends to the quotient upon adding the term $\delta (e e^\dag_n \xi^n)$. In the local chart defined by \eqref{inproofProjectionbulktoboundary}, we define:\footnote{Since $\delta \epsilon_n=0$ we obtain, in the space of preboundary fields, $\delta \varepsilon e_n = -(1+\varepsilon)\delta e_n$. Hence $\epsilon_n \delta \txi{n}= (1+ \varepsilon)e_n \delta (1+\varepsilon)^{-1} \xi^n+ e_n \delta \xi^n=-(1+ \varepsilon)e_n  (1+\varepsilon)^{-2} \delta \varepsilon \xi^n+ e_n \delta \xi^n= \delta e_n \xi^n+e_n \delta \xi^n$.} 
 \begin{align*}
\alpha^{\filt{1}}:= \trintl{\filtint{1}} - \te \delta \tom+ \tom^\dag \delta \tc - \tilde{f_n}^{\dag} \epsilon_n \delta \txi{n}- \iota_{\delta \txi{}} \te \tilde{f_n}^{\dag}- \iota_{\delta \txi{}} \te \tilde{f_a}^{\dag}+ \iota_{\txi{}}\tom^\dag \delta \tom,
\end{align*}
so that $\check{\alpha}^{\filt{1}}+\delta (e e^\dag_n \xi^n)= \pi_{\mathrm{red}}^{\filt{1}*} \alpha^{\filt{1}}$ and, consequently,
 \begin{align*}
\varpi^{\filt{1}}:= \trintl{\filtint{1}} &- \delta\te \delta \tom+  \delta\tom^\dag \delta \tc - \delta\tilde{f_n}^{\dag} \epsilon_n \delta \txi{n}+ \iota_{\delta \txi{}} \delta(\te \tilde{f_n}^{\dag}) \\ &+ \iota_{\delta \txi{}}\delta( \te \tilde{f_a}^{\dag})+ \delta(\iota_{\txi{}}\tom^\dag )\delta \tom.
\end{align*}

{

Given any element $k \in \Omega^2(\partial M, \wedge^2 \mathcal{V})$ and a basis $\{v_i\}_{i=1,2,3}$ of the vector space $\mathcal{V}$ we can build a diffeomorphism  
\begin{align*}
\phi : \Omega^2(\partial M, \wedge^2 \mathcal{V}) &\rightarrow \prod_{i=1}^3 p_i^{\dag}\Omega^2(\partial M, \wedge^2 \mathcal{V})\\
k  & \mapsto \left( p_1^{\dag}k,  p_2^{\dag} k,  p_3^{\dag} k \right)
\end{align*}
where $p_i^{\dag}$ is the projection defined by \eqref{projection_p_i^dag} with inverse given by $\left( p_1^{\dag} k,  p_2^{\dag} k,  p_3^{\dag}k \right) \mapsto  p_1^{\dag} k + p_2^{\dag} k + p_3^{\dag} k = k$.
We now observe that the quantities $\tilde{f_a}^{\dag}$ and $\tilde{f_n}^{\dag}$ in \eqref{inproofProjectionbulktoboundary} satisfy $ p_a^{\dag}\tilde{f_a}^{\dag}=\tilde{f_a}^{\dag}$ and $ p_n^{\dag}\tilde{f_n}^{\dag}=\tilde{f_n}^{\dag}$ where  $p_a^{\dag}$ and $p_n^{\dag}$ are with respect to the basis $(e_a, e_n)$. Hence defining 
$\te^\dag := \tilde{f_a}^{\dag} + \tilde{f_n}^{\dag}$ we obtain that the image of $\te^\dag$ under the diffeomorphism $\phi$ is the couple  $\left (\tilde{f_a}^{\dag},\tilde{f_n}^{\dag}\right)$. 
Since the codimension-$1$ quantities depend only on the $\te^\dag$ as defined above, we conclude that we can use it as a basis independent field on $\filtint{1}$.

} Hence the odimension-$1$ forms can be more conveniently rewritten as
 \begin{align*}
\alpha^{\filt{1}}& := \trintl{\filtint{1}} - \te \delta \tom+ \tom^\dag \delta \tc -\te^{\dag} \epsilon_n \delta \txi{n}- \iota_{\delta \txi{}} \te \te^{\dag}+ \iota_{\txi{}}\tom^\dag \delta \tom\\
\varpi^{\filt{1}} &:= \trintl{\filtint{1}} - \delta\te \delta \tom+  \delta\tom^\dag \delta \tc - \delta\te^{\dag} \epsilon_n \delta \txi{n}+ \iota_{\delta \txi{}} \delta(\te \te^{\dag})+ \delta(\iota_{\txi{}}\tom^\dag )\delta \tom.
\end{align*}
{
where $\te\in \mathcal{U}^{\filt{1}}$, $\tom\in\Omega^1(\filtint{1},\wedge^2\mathcal{V})$, $\tc\in\Omega^0[1](\filtint{1},\wedge^2\mathcal{V})$, $\txi{}\in\mathfrak{X}[1](\filtint{1})$, $\txi{n}\in C^\infty[1](\filtint{1})$, $\tom^\dag\in\Omega^3[-1](\filtint{1},\mathcal{V})$ and $\te^\dag\in\Omega^{2}[-1](\filtint{1},\wedge^2 \mathcal{V})$.

Observe that, due to the arbitrariness of $\epsilon_n$, and since all expressions involving $\te_a$ are coordinate independent on $\filtint{1}$, the above discussion ensures the independence of $\alpha^{\filt{1}}$ from the choice of a specific set of coordinates.}

Finally, we can compute $\check{S}^{\filt{1}}$ such that  $\iota_{Q^{\filt{0}}} \iota_{Q^{\filt{0}}} \varpi^{\filt{0}} = 2 \check{\pi}^{\filt{0}*} \check{S}^{\filt{1}}$, using Equation \eqref{Q_GR-bulk1} for the bulk $Q$. A straightforward calculation yields: 

\begin{align*}
\check{S}^{\filt{1}}= \trintl{\filtint{1}}&  - \iota_\xi e F_{\omega} -e_n \xi^n F_\omega  - c d_\omega e + \frac12 [c,c]\omega^\dag  + c [e_n^\dag,e]\xi ^n+c [e^\dag,e_n]\xi ^n \\
&+ c d_\omega( \iota_\xi \omega^\dag + \omega^\dag_n \xi^n) -c d_\omega(\iota_\xi c^\dag_n \xi^n) + \iota_\xi F_\omega \iota_\xi c^\dag_n \xi^n - \frac12 \iota_\xi \iota_\xi F_\omega\omega^\dag \\ &-  \iota_\xi F_\omega\omega^\dag_n \xi^n - \iota_\xi F_\omega\iota_\xi\omega^\dag+ \frac12 \iota_{[\xi,\xi]}e e^\dag + e_n \xi^n d_{\omega} \iota_\xi e^\dag+ \iota_\xi e d_{\omega}(e^\dag_n \xi^n)  \\ &+ e^\dag_n \xi^n \iota_\xi d_{\omega} e + e_n \xi^n d_{\omega}(e^\dag_n \xi^n) - \frac12 \iota_{[\xi,\xi]}\chi\mathsf{V}\xi^n - \iota_\xi d \xi^n \chi_n\mathsf{V}\xi^n.
\end{align*}
Using the projection \eqref{Projectionbulktoboundary} we can find the codimension-$1$ action functional to be
\begin{align*}
S^{\filt{1}}= \trintl{\filtint{1}}& - \iota_{\txi{}} \te F_{\tom} -\epsilon_n \txi{n} F_{\tom} - \tc d_{\tom} \te + \frac12 [\tc,\tc]\tom^{\dag} +\frac12 \iota_{\txi{}} \iota_{\txi{}} F_{\tom}\tom^\dag+ \frac12 \iota_{[\txi{},\txi{}]}\te\te^\dag \nonumber \\ &+\tc d_{\tom}( \iota_{\txi{}} \tom^\dag)+ L_{\txi{}}^{\tom} (\epsilon_n \txi{n}) \te^\dag- [\tc, \epsilon_n \txi{n}]\te^\dag
\end{align*}
satisfying $\check{S}^{\filt{1}}= \pi^{\filt{1}*}_{\text{red}} S^{\filt{1}}$. As a direct consequence, then, we have
$$
\iota_{Q^{\filt{0}}}\iota_{Q^{\filt{0}}} \varpi^{\filt{0}} = 2\pi^{\filt{0}*}S^{\filt{1}}.
$$

Having found the codimension-$1$ action functional we can compute the cohomological vector field $Q^{\filt{1}}$ on the $1$-stratum. Since the coordinates we are using are not a Darboux chart, some complications in the computation arise. Nonetheless, the non-degeneracy of $\varpi^{\filt{1}}$ guarantees that starting from the variation of the action and using the equation $\iota_{Q^{\filt{1}}}\varpi^{\filt{1}}= \delta S^{\filt{1}}$, we can compute the cohomological codimension-$1$ vector field $Q^{\filt{1}}$.
\end{proof}

\begin{corollary}[to the proof of Proposition \ref{prop:BFVdata}]\label{cor:BFVdata-projection}
An explicit expression of the projection to codimension-$1$ fields, given a coordinate system adapted to the embedding $\filtint{1}\to\filtint{0}$ is given by
\begin{equation}\label{Projectionbulktoboundary}
{\pi}^{\filt{0}}_{GR}:
\begin{cases}
\te=  e - \omega^\dag_n\xi^n +\iota_\xi c^\dag_n\xi^n\\
\tom= \omega - e^\dag_n\xi^n\\
\tc=  c -  \iota_\xi e^\dag_n\xi^n \\
\txi{n} = (1+\varepsilon)^{-1}\xi^n\\
\txi{a} = \xi^a \\
\te^\dag = e^{\dag} - \chi_a\xi^n\mathsf{V}^{\comp{a}}  +\left(\omega_{na}^\dag\xi^n e^\dag \right)^{\comp{a}}- (e^\dag_n\xi^n\omega^\dag_a)^{\comp{a}} - \left(\iota_\xi c^\dag_{na}\xi^n e^\dag\right)^{\comp{a}} \\
\qquad - \chi_n\xi^n\mathsf{V}^{\comp{n}}  + \left(\omega_n^\dag\xi^n e^\dag_n \right)^{\comp{n}}- (e^\dag_n \iota_\xi c^\dag_n\xi^n)^{\comp{n}}\\
\tom^\dag=\omega^\dag\\
\epsilon_n= (1+\varepsilon)e_n
\end{cases}\hspace{-0.3cm};
\end{equation}
where $\mathsf{V} \in \Omega^3(\partial M, \wedge^3 \mathcal{V})$ is a fixed volume form, $\varepsilon \in C^\infty(M^{\filt{1}})$ is a smooth function, $\chi \in \Omega(M)$ is the one form component of $\xi^{\dag}= \chi \otimes \mathsf{V}$,  and we used the notation explained in \eqref{componenttopform},  \eqref{restriction} and \eqref{vectorboundary} for the restriction of field and their normal components where the superscript ${}^{\comp{n}}$ here denotes the component with respect to $e_n$.
\end{corollary}

\begin{remark}
Observe that, strictly speaking, Equation \eqref{rCME} is satisfied by the above data only if we modify $S_{GR}$ by the boundary term \emph{Tr}$\int_{\filtint{0}} d(e e_n^\dag \xi^n)$, so that the associated \emph{pre-boundary one-form} $\check{\alpha}^{\filt{1}}$ is automatically basic. Indeed, this is necessary only if we insist on the BV-BFV data to be \emph{exact}, i.e. such that the symplectic forms are exact at every codimension $\varpi^{\filt{k}}=\delta\alpha^{\filt{k}}$ \emph{and} that their symplectic potential $\alpha^{\filt{k}}$ is pulled back to the respective term in \eqref{rCME}. A picture suitable to situations like the present one, where symplectic reduction is possible, but nontrivial, at every codimension, is to consider the symplectic forms instead of their potentials, and the equation
$$
\mathcal{L}_Q \varpi^{\filt{k}} = \pi^{\filt{k}*}\varpi^{\filt{k+1}},
$$
which follows from \eqref{rCME} by differentiating w.r.t. $\delta$. However, modifying $S^{\filt{k}}$ by a term concentrated in codimension-$(k+1)$ does not change the BV-BFV structure (cf. \cite{MSW2019}).
\end{remark}

\begin{remark}
The expression \eqref{BoundarytwoformGR} is not in the Darboux form. The change of coordinates to a Darboux chart will turn out to coincide with the boundary symplectomorphism between GR and BF theory (see Section \ref{Sect:BCV}). The same symplectomorphism will also turn the local trivialisation \eqref{LoctrivF1} into a global one, showing that the bundle $\mathcal{F}^{\filt{1}}_{GR}$ is indeed trivial (since $\mathcal{F}^{\filt{1}}_{BF}$ is). 
\end{remark}

\subsection{2-extended GR theory }
We are now ready to compute the structure induced on codimension-$2$ strata when $M$ carries a $2$-stratification, for example in the presence of corners. Building up from the codimension-$1$ BV-BFV structure found in Proposition \ref{prop:BFVdata}, denoting again by $\Omega_{nd}^1(\filtint{2}, \mathcal{V})$ the space of maps $e\colon T\filtint{2} \longrightarrow \mathcal{V}$ whose image defines a linearly independent system, we have the following result.

\begin{proposition}
\label{prop:BFFVdata}
The BV theory $\mathfrak{F}^{\uparrow 0}_{GR}$ is 2-extendable to $\mathfrak{F}^{\uparrow 2}_{GR}$. The codimension-$2$ data are: 
\begin{itemize}
\item The space of codimension-$2$ fields, given by the bundle
\begin{equation}
\mathcal{F}^{\filt{2}}_{GR}\longrightarrow \Omega_{nd}^1(\filtint{2}, \mathcal{V}),
\end{equation}
with local trivialisation on a open subset $\mathcal{U}^{\filt{2}} \subset \Omega_{nd}^1(\filtint{2}, \mathcal{V})$
\begin{equation*}
\mathcal{F}^{\filt{2}}_{GR}=\mathcal{U}^{\filt{2}} \times\Omega^1( \filtint{2}, \wedge^2\mathcal{V})\oplus \Omega^0[1]( \filtint{2}, \wedge^2\mathcal{V})\oplus \mathfrak{X}[1](\filtint{2}) \oplus C^\infty[1](\filtint{2})^{\oplus 2},
\end{equation*}
and fields denoted by $\tte\in \mathcal{U}^{\filt{2}}$, $\ttom\in\Omega^1(\filtint{2},\wedge^2\mathcal{V})$ in degree zero, $\ttc\in\Omega^0[1](\filtint{2},\wedge^2\mathcal{V})$, $\ttxi{}\in\mathfrak{X}[1](\filtint{2})$ and $\ttxi{m},\ttxi{n}\in C^\infty[1](\filtint{2})$ in degree one,
together with two linearly independent, fixed vector fields $\epsilon_m, \epsilon_n \in \Gamma(\mathcal{V})$, completing the image of elements $\tte\in \mathcal{U}^{\filt{2}} \subset \Omega_{nd}^1(\filtint{2}, \mathcal{V})$ to a basis of  $\mathcal{V}$;

\item The codimension-$2$ one-form, symplectic form and action functional
\begin{subequations} \begin{align}
\alpha^{\filt{2}}_{GR}& = \trintl{\filtint{2}}  \ttc \delta \tte+\iota_{\ttxi{}} \tte \delta \ttom + \epsilon_m \ttxi{m}\delta \ttom +\epsilon_n \ttxi{n}\delta \ttom,\\
\varpi^{\filt{2}}_{GR}& = \trintl{\filtint{2}} \delta \tte \delta \ttc +\iota_{\delta \ttxi{}} \tte \delta \ttom+\iota_{\ttxi{}} \delta \tte \delta \ttom + \epsilon_m \delta \ttxi{m}\delta \ttom +\epsilon_n \delta \ttxi{n}\delta \ttom,\\
S^{\filt{2}}_{GR} & = \trintl{\filtint{2}} -\frac12 [\ttc,\ttc] \tte - \iota_{\ttxi{}} \tte d_{\ttom} \ttc- \epsilon_m \ttxi{m} d_{\ttom} \ttc - \epsilon_n \ttxi{n} d_{\ttom} \ttc,
\end{align}\end{subequations}
\item The cohomological vector field $Q^{\filt{2}}_{GR}$
\begin{subequations}\label{Q_GR-corner1}
\begin{align}
 Q^{\filt{2}}_{GR} {\tte} &=- d_{\ttom} (e_a \ttxi{a})-d_{\ttom} (\epsilon_m \ttxi{m})-d_{\ttom} (\epsilon_n \ttxi{n})-[\ttc,\tte]\\
 Q^{\filt{2}}_{GR} {\ttom} &= -d_{\ttom} \ttc \\
 Q^{\filt{2}}_{GR} {\ttc} &= -\frac{1}{2}[\ttc, \ttc]\\
 Q^{\filt{2}}_{GR} {\ttxi{}} &=  \frac{1}{2}[\txi{}, \txi{}] + X_n^{\comp{a}} -Y_n^{\comp{a}}+X_m^{\comp{a}} -Y_m^{\comp{a}}\\
 Q^{\filt{2}}_{GR} {\ttxi{n}} &= X_n^{\comp{n}} -Y_n^{\comp{n}}+X_m^{\comp{n}} -Y_m^{\comp{n}}\\
 Q^{\filt{2}}_{GR} {\ttxi{m}} &= X_n^{\comp{m}} -Y_n^{\comp{m}}+X_m^{\comp{m}} -Y_m^{\comp{m}} 
\end{align}
\end{subequations}
where 
$$X_n= L_{\ttxi{}}^{\ttom}(\epsilon_n \ttxi{n}), \qquad X_m= L_{\ttxi{}}^{\ttom}(\epsilon_m \ttxi{m}),  \qquad Y_n= [\ttc, \epsilon_n \ttxi{n}], \qquad Y_m= [\ttc, \epsilon_m \ttxi{m}],$$ 
and the superscripts ${}^{\comp{n}},{}^{\comp{m}}$ denote components with respect to $\epsilon_n,\epsilon_m$  as defined in Section \ref{sec:Notation};
\item The projection to codimension-$2$ fields ${\pi}^{\filt{1}}_{GR} = {\pi_{\mathrm{red}}}^{\filt{2}} \circ {\check{\pi}}^{\filt{1}}$ is smooth, where ${\check{\pi}}^{\filt{1}}: \mathcal{F}^{\filt{1}}\to \check{\mathcal{F}}^{\filt{2}}$ is the restriction of the codimension-$1$ fields to the $2$-stratum $M^{\filt{2}}$ and ${\pi_{\mathrm{red}}}^{\filt{2}}: \check{\mathcal{F}}^{\filt{2}}\to \mathcal{F}^{\filt{2}}$ is the symplectic reduction \eqref{symplectic-reduction}.
\end{itemize}
\end{proposition}

\begin{proof}
We proceed as before, following the strategy outlined in Section \ref{sec:Notation}. Let $\check{\alpha}^{\filt{2}}$ be such that $\iota_{Q^{\filt{1}}} \varpi^{\filt{1}}= \delta S^{\filt{1}}+ \check{\pi}^{\filt{1}}\check{\alpha}^{\filt{2}}$:
\begin{align}\label{check alpha boundary}
\check{\alpha}^{\filt{2}}= \trintl{\filtint{2}}& \iota_{\txi{}} \te \delta \tom+\te_m \txi{m} \delta \tom+ \epsilon_n \txi{n} \delta \tom+ \tc\delta \te - \delta \tom \iota_{\txi{}} \tom^\dag_m \txi{m}  +  \iota_{\delta\txi{}}\te \te^\dag_m \txi{m}\\
&-\te_m \delta \txi{m} \te^\dag_m \txi{m}-\tc \delta (\tom_m^\dag \txi{m})- \epsilon_n \delta \txi{n} \te^{\dag}_m \txi{m} \nonumber
\end{align}
\begin{align*}
\check{\varpi}^{\filt{2}}= \delta  \check{\alpha}^{\filt{2}}= \trintl{\filtint{2}} &\iota_{\delta \txi{}} \te \delta \tom+\iota_{\txi{}} \delta\te \delta \tom+\delta \te_m \txi{m} \delta \tom-\te_m \delta \txi{m} \delta \tom- \epsilon_n \delta \txi{n} \delta \tom+ \delta \tc\delta \te\\
&-\delta \tom \iota_{\delta\txi{}} \tom^\dag_m \txi{m}- \delta \tom \iota_{\txi{}} \delta \tom^\dag_m \txi{m} +\delta \tom \iota_{\txi{}} \tom^\dag_m \delta \txi{m} -\delta \te_m \delta \txi{m} \te^\dag_m \txi{m} \\ 
&+\te_m \delta \txi{m} \delta \te^\dag_m \txi{m}+\te_m \delta \txi{m} \te^\dag_m \delta \txi{m}- \iota_{\delta\txi{}}\delta \te \te^\dag_m \txi{m}-  \iota_{\delta\txi{}}\te \delta \te^\dag_m \txi{m}\\&
-   \iota_{\delta\txi{}}\te \te^\dag_m \delta \txi{m}-\delta \tc \delta (\tom_m^\dag \txi{m})+\epsilon_n \delta \txi{n} \delta \te^{\dag}_m \txi{m}+ \epsilon_n \delta \txi{n} \te^{\dag}_m \delta \txi{m}
\end{align*}
The forms $\check{\alpha}^{\filt{2}},\check{\varpi}^{\filt{2}}$ are defined on $\check{\mathcal{F}}^{\filt{2}}$, the space of restrictions of codimension-$1$ fields (and their normal jets) to the codimension-$2$ stratum, with $\check{\pi}^{\filt{1}}\colon \mathcal{F}^{\filt{1}} \longrightarrow \check{\mathcal{F}}^{\filt{2}}$.

We have to show that the symplectic reduction of $\check{\varpi}^{\filt{2}}$ is possible. The equations that define the kernel of $(\check{\varpi}^{\filt{2}})^\sharp$ are:
\begin{subequations}\begin{align}
\label{kernel_precorner_i} \delta \tc: & \qquad X_{\te}+(X_{ \tom^\dag_m}) \txi{m}-  \tom^\dag_m( X_{\txi{m}})=0\\
\label{kernel_precorner_a} \delta \tom: & \qquad \iota_{(X_{ \txi{}})} \te+\iota_{\txi{}} X_{\te}+X_{ \te_m} \txi{m}- \te_m (X_{\txi{m}})- \epsilon_n (X_{\txi{n}})\\ & \qquad \qquad+\iota_{(X_{\txi{}})} \tom^\dag_m \txi{m} +\iota_{\txi{}}(X_{ \tom^\dag_m}) \txi{m}-\iota_{\txi{}} ( \tom^\dag_m) X_{\txi{m}}=0 \nonumber\\
\label{kernel_precorner_d} \delta \txi{}: & \qquad -\te_\bullet p^\dag_p(X_{\tom})-(X_{\tom})\tom^\dag_{m \bullet} \txi{m}+ \te_\bullet (X_{\te^\dag_m}) \txi{m} + \te_\bullet \te^\dag_m (X_{\txi{m}})  =0\\
\label{kernel_precorner_g} \delta \txi{n} : & \qquad - \epsilon_n p^\dag_n(X_{\tom})+ \epsilon_n {\te^{\dag}}_m(X_{\txi{m}})+ \epsilon_n(X_{{\te^{\dag}}_m})\txi{m}=0\\
\label{kernel_precorner_h} \delta \txi{m}: & \qquad - \te_m p^\dag_m(X_{\tom})+(X_{\tom})\iota_{\txi{}}\tom^\dag_m - \iota_{(X_{\txi{}})} \te \te_m	^\dag -2e_m{\te^{\dag}}_m(X_{\txi{m}})\\ & \qquad -  (X_{\te_m}){\te^\dag}_m\txi{m}+\te_m (X_{{\te^{\dag}}_m})\txi{m}+(X_{\tc}) \tom^\dag_m +\epsilon_n (X_{\txi{n}}) \te^\dag_m=0 \nonumber \\
\label{kernel_precorner_f} \delta \te_m : & \qquad \txi{m} (X_{\tom})+ (X_{\txi{m}}) \te^\dag_m \txi{m} =0\\
\label{kernel_precorner_e} \delta\te: & \qquad \iota_{\txi{}} (X_{\tom}) + (X_{\tc})=0\\
\label{kernel_precorner_l} \delta \tom_m^\dag : & \qquad - \iota_{\txi{}} (X_{\tom}) \txi{m}-(X_{\tc})\txi{m} =0\\ 
\label{kernel_precorner_m} \delta \te^\dag_m: & \qquad \iota_{(X_{ \txi{}})} \te \txi{m}- \te_m (X_{\txi{m}}) \txi{m}- \epsilon_n (X_{\txi{n}}) \txi{m}=0
\end{align}\end{subequations}

From \eqref{kernel_precorner_i} we get 
$$X_{\te}=-(X_{ \tom^\dag_m}) \txi{m}+  \tom^\dag_m (X_{\txi{m}}).$$
Inserting this result into \eqref{kernel_precorner_a} we get
\begin{align*}
\iota_{(X_{ \txi{}})} \te+X_{ \te_m} \txi{m}- \te_m (X_{\txi{m}})- \epsilon_n (X_{\txi{n}})+\iota_{(X_{\txi{}})} \tom^\dag_m \txi{m}=0
\end{align*}
which is a vector equation. {
Using the basis $\{\te, \te_m, \epsilon_n\}$ we can write it equivalently as 
\begin{align*}
&\te^a X_{ \txi{a}}
+(X_{e_m}^{\comp{a}})\te_a \txi{m} +(X_{e_m}^{\comp{m}})\te_m \txi{m} +(X_{e_m}^{\comp{n}})\epsilon_n \txi{m}
- \te_m (X_{\txi{m}})
- \epsilon_n (X_{\txi{n}})\\
&+ \tom^{\dag\comp{a}}_{mb} \te_a X_{\txi{b}}\txi{m}
+ \tom^{\dag\comp{m}}_{mb} \te_m X_{\txi{b}}\txi{m}
+ \tom^{\dag\comp{n}}_{mb} \epsilon_n X_{\txi{b}}\txi{m}
=0.
\end{align*}
In matrix notation we get
\begin{equation}\label{matrixequationcorner}
\left(\begin{array}{ccc}
1 -  \tom^{\dag\comp{a}}_{mb} \txi{m} & 0 & 0 \\
-\tom^{\dag\comp{m}}_{mb} \txi{m} & 1 &  0\\
-\tom^{\dag\comp{n}}_{mb} \txi{m} & 0 & 1
\end{array}\right)
\left(\begin{array}{c}
(X_{\txi{b}})\\
(X_{\txi{m}})\\
(X_{\txi{n}})
\end{array}\right)=
\left(\begin{array}{c}
(X_{e_m}^{\comp{a}})\\
(X_{e_m}^{\comp{m}})\\
(X_{e_m}^{\comp{n}})
\end{array}\right)\txi{m}.
\end{equation}
}
We can write equation \eqref{matrixequationcorner} as $(1+Y)x = b$ where $Y$ is a nilpotent matrix of index 2. Hence the equation can be inverted by $ x = (1 - Y + Y^2) b$ and we get:
\begin{equation}
\left(\begin{array}{c}
(X_{\txi{a}})\\
(X_{\txi{m}})\\
(X_{\txi{n}})
\end{array}\right)=\left(\begin{array}{ccc}
1 +  \tom^{\dag\comp{a}}_{ma} \txi{m} & 0 & 0 \\
+\tom^{\dag\comp{m}}_{ma} \txi{m} & 1 &  0\\
+\tom^{\dag\comp{n}}_{ma} \txi{m} & 0 & 1
\end{array}\right)
\left(\begin{array}{c}
(X_{e_m}^{\comp{a}})\\
(X_{e_m}^{\comp{m}})\\
(X_{e_m}^{\comp{n}})
\end{array}\right)\txi{m}.
\end{equation}
Since $(\txi{m})^2=0$ we get 
$$(X_{\txi{\mu}})=(X_{e_m}^{\comp{\mu}}) \txi{m} \qquad \mu= a,m,n,$$
and, for the same reason, \eqref{kernel_precorner_m} is satisfied.
From \eqref{kernel_precorner_g} we get
\begin{equation}\label{pnX_tom}
 p^\dag_n(X_{\tom})=  p^\dag_n {\te^{\dag}}_m(X_{e_m}^{\comp{m}}) \txi{m}+ p^\dag_n(X_{{\te^{\dag}}_m})\txi{m}.
\end{equation}
Equation \eqref{kernel_precorner_e} is easily solved as 
$$ (X_{\tc})= -\iota_{\txi{}} (X_{\tom})$$
which in turn solves also \eqref{kernel_precorner_l}.
Using equation \eqref{kernel_precorner_e} we can solve \eqref{kernel_precorner_h}:
\begin{align}\label{pmX_tom}
 p^\dag_m(X_{\tom})= p^\dag_m{\te^{\dag}}_m (X_{e_m}^{\comp{m}}) \txi{m} + p^\dag_m(X_{{\te^{\dag}}_m})\txi{m}.
\end{align}
Turning to 	\eqref{kernel_precorner_d} we get
$$p^\dag_a(X_{\tom})(1+ \tom^{\dag\comp{a}}_{m} \txi{m})=  p^\dag_a (X_{\te^\dag_m}) \txi{m} + p^\dag_a \te^\dag_m (X_{e_m}^{\comp{m}}) \txi{m} ,$$
hence
\begin{equation}\label{paX_tom}
p^\dag_a(X_{\tom})=  p^\dag_a (X_{\te^\dag_m}) \txi{m} + p^\dag_a \te^\dag_m (X_{e_m}^{\comp{m}}) \txi{m} .
\end{equation}
Collecting \eqref{pnX_tom}, \eqref{pmX_tom} and \eqref{paX_tom}, we get that $(X_{\tom})$ is proportional to $\txi{m}$, hence also \eqref{kernel_precorner_f} is solved.
This shows that it is possible to perform symplectic reduction.

Collecting all remaining nontrivial equations  together we obtain
\begin{align*}
(X_{\txi{\mu}})&=(X_{e_m}^{\comp{\mu}}) \txi{m} \qquad \mu= a,m,n \\
(X_{\tom})&= (X_{\te^\dag_m}) \txi{m} +  \te^\dag_m (X_{e_m}^{\comp{m}}) \txi{m}\\
(X_{\tc})&= \iota_{\txi{}} (X_{\te^\dag_m}) \txi{m} + \iota_{\txi{}} \te^\dag_m (X_{e_m}^{\comp{m}}) \txi{m}\\
 (X_{\te})&=-(X_{ \tom^\dag_m}) \txi{m}+  \tom^\dag_m(X_{e_m}^{\comp{m}}) \txi{m}.
\end{align*}
The vertical vector fields are
\begin{align*}
\mathbb{E}_m =& \X{e_m}\pard{}{\te_m} + \te^\dag_m (X_{e_m}^{\comp{m}}) \txi{m}\pard{}{\tom} + \iota_{\txi{}} \te^\dag_m (X_{e_m}^{\comp{m}}) \txi{m}\pard{}{\tc}\\\notag
 +& \tom^\dag_m(X_{e_m}^{\comp{m}}) \txi{m}\pard{}{e} - \X{e_m}^\mu\txi{m}\pard{}{\txi{\mu}}\\
 \mathbb{E}^\dag_m =& \X{\te^\dag_m}\pard{}{\te^\dag_m} +(X_{\te^\dag_m}) \txi{m} \pard{}{\tom} + \iota_{\txi{}}\X{\te^\dag_m}\txi{m} \pard{}{\tc} \\
\mathbb{\Omega}^\dag_n =& \X{\tom^\dag_m} \pard{}{\tom^\dag_m} - \X{\tom^\dag_m}\txi{m} \pard{}{\te}.
\end{align*}
With an analogous procedure to Proposition \ref{prop:BFVdata} we flow along these vertical vector fields to obtain the new corner variables.
Using $\mathbb{\Omega}^\dag_m$ we have that $\tom_m^\dag(s=1)=0\iff \X{\tom^\dag_m}=-\tom^\dag_m$ which leads to
$$\te{[1]} = \te - \tom^\dag_m\txi{m}.$$
Analogously, with $\mathbb{E}^\dag_m$ we get
\begin{eqnarray}
\tom{[1]}=\tom -\te^\dag_m\txi{m}; & \tc{[1]}=\tc - \iota_{\txi{}} \te^\dag_m\txi{m}.\nonumber
\end{eqnarray}	
Now we have to consider 
$$\mathbb{E}_m = \X{e_m}\pard{}{\te_m} - \X{e_m}^\mu\txi{m}\pard{}{\txi{\mu}}.$$
As before, this sets the vector $\te_m$ to a constant $\epsilon_m= (1 + \varepsilon ') \te_m$ { for $\varepsilon ' \in C^\infty(\filtint{1})$, $\varepsilon '>0$,} and transforms $\txi{m}[1]= (1+\varepsilon ')^{-1}\txi{m}$, while leaving $\txi{a}$ and $\txi{n}$ unchanged. Once again, fixing a linearly independent vector $\epsilon_m$ can be done in an open subset $\mathcal{U}^{\filt{2}}\subset \Omega_{nd}^1(\filtint{2}, \mathcal{V})$ independently of $\te$. Summarizing, if we denote the space of codimension-$2$ fields by
$$
\mathcal{F}^{\filt{2}} \coloneqq \qsp{\check{\mathcal{F}}^{\filt{2}}}{\mathrm{ker}(\check{\varpi}^{\filt{2}}{}^\sharp)}
$$
we then obtain the following projection $\pi_{\mathrm{red}}^{\filt{2}}:  \check{\mathcal{F}}^{\filt{2}}\rightarrow  \mathcal{F}^{\filt{2}}$:
\begin{equation}
\pi_{\mathrm{red}}^{\filt{2}}:
\begin{cases}
\tte := \te-\tom^\dag_m \txi{m}\\
\ttom:= \tom -\te^\dag_m \txi{m}\\
\ttc:= \tc - \iota_{\txi{}} \te^\dag_m \txi{m} \\
\ttxi{a}:= \txi{a} \\
\ttxi{n}:= \txi{n} \\
\ttxi{m}:= (1+\varepsilon ')^{-1}\txi{m}
\end{cases}
\end{equation}
together with $\epsilon_m= (1 + \varepsilon ') \te_m$, and once more we define the BV-BFV map to be ${\pi}^{\filt{1}}\coloneqq \pi_{\mathrm{red}}^{\filt{2}}\circ \check{\pi}^{\filt{1}}$.

The one-form $\check{\alpha}^{\filt{2}}$ is not basic w.r.t.  $\pi_{\mathrm{red}}^{\filt{2}}$, and it descends to the quotient only upon adding the exact term $\delta ( \iota_{\txi{}}\te \te^\dag_m \txi{m} + \epsilon_n \txi{n}\te^\dag_m \txi{m})$. The codimension-$2$ one-form is then given by 
$$\alpha^{\filt{2}}= \trintl{\filtint{2}}  \ttc \delta \tte+\iota_{\ttxi{}} \tte \delta \ttom + \epsilon_m \ttxi{m}\delta \ttom +\epsilon_n \ttxi{n}\delta \ttom.$$

 From the defining formula $ \iota_{Q^{\filt{1}}}\iota_{Q^{\filt{1}}} \varpi^{\filt{1}} = 2 \check{\pi}^{\filt{1}*} \check{S}^{\filt{2}}$ and the expression for $Q^{\filt{1}}$ \eqref{Q_GR-boundary1}, we compute the \textit{pre-codimension-$2$} action functional $\check{S}^{\filt{2}}$. 
\begin{align*}
\check{S}^{\filt{2}}=\trintl{\filtint{2}}& \te_m \txi{m} d_{\tom} c +\epsilon_n \txi{n} d_{\tom} c+ \iota_{\txi{}} \te d_{\tom} \tc -\frac{1}{2}[\tc,\tc]\te  +\iota_{\txi{}} \tom^\dag_m \txi{m}   d_{\tom} \tc
 \\&+ \frac{1}{2}[\tc,\tc]\tom^\dag_m \txi{m}+[\tc, \epsilon_n 	\txi{n}] \te_m^\dag \txi{m}- \iota_{\txi{}} d_{\tom} (\epsilon_n \txi{n})  \te_m^\dag \txi{m}
 \\&-\iota_{\txi{}}d \txi{m}\te_m\te_m^\dag \txi{m}-\iota_{\txi{}}d \txi{a}\te_a\te_m^\dag \txi{m}.
\end{align*}
It is easy to show that $\check{S}^{\filt{2}}$ is basic, and, defining
\begin{align*}
S^{\filt{2}}= \trintl{\filtint{2}} -\frac12 [\ttc,\ttc] \tte + \iota_{\ttxi{}} \tte d_{\ttom} \ttc+ \epsilon_m \ttxi{m} d_{\ttom} \ttc+ \epsilon_n \ttxi{n} d_{\ttom} \ttc
\end{align*}
 we obtain that $\pi_{\mathrm{red}}^{\filt{2}*}S^{\filt{2}}=\check{S}^{\filt{2}}$ and
$$
\iota_{Q^{\filt{1}}}\iota_{Q^{\filt{1}}}\varpi^{\filt{1}} = 2 \pi^{\filt{1} *}S^{\filt{2}}.
$$
\end{proof}

\begin{corollary}[to the proof of Proposition \ref{prop:BFFVdata}]\label{cor:BFFVdata-projection}
The projection to codimension-$2$ fields, given a coordinate system adapted to the embedding $\filtint{2}\to\filtint{1}$ (see Remark \ref{rem:coordinatedependingquantities}): 
\begin{equation}
\pi^{\filt{1}}_{GR}:
\begin{cases}
\tte := \te-\tom^\dag_m \txi{m}\\
\ttom:= \tom -\te^\dag_m \txi{m}\\
\ttc:= \tc - \iota_{\txi{}} \te^\dag_m \txi{m} \\
\ttxi{a}:= \txi{a} \\
\ttxi{n}:= \txi{n} \\
\ttxi{m}:= (1+\varepsilon ')^{-1}\txi{m} \\
\epsilon_m: = (1 + \varepsilon ') \te_m
\end{cases};
\end{equation}
where $\varepsilon '\in C^\infty(M^{\filt{2}})$ and we used the notation explained in \eqref{restriction} and \eqref{vectorboundary} for the restriction of field and their tranversal components, adapted as described in  Remark \ref{rem:notationhighercodimension}.
\end{corollary}

\subsection{3-extended GR theory }
Finally, we allow $M$ to bear a $3$-stratification, i.e. we consider $\{M^{\filt{k}}\}_{k=0\dots 3}$. We iterate once again the BV-BFV procedure from the previously obtained, codimension-$2$ data.  
\begin{proposition}
\label{prop:BFFFVdata}
The BV theory $\mathfrak{F}^{\uparrow 0}_{GR}$ is 3-extendable to $\mathfrak{F}^{\uparrow 3}_{GR}$. The codimension-$3$ data are: 
\begin{itemize}
\item The space of fields
\begin{equation*}
\mathcal{F}^{\filt{3}}_{GR}= \Omega^0[1](M^{\filt{3}},\wedge^2 \mathcal{V}) \oplus C^\infty[1](M^{\filt{3}})^{\oplus 3}
\end{equation*}
together with a basis of $\mathcal{V}$ denoted by $\{\epsilon_m, \epsilon_n, \epsilon_a\}$;

\item The codimension-$3$ one-form, symplectic form and action functional
\begin{align}
\alpha^{\filt{3}}_{GR}& = \trintl{\filtint{3}} -\epsilon_n \widetilde{\ttxi{n}} \delta \widetilde{\ttc} -\epsilon_m \widetilde{\ttxi{m}} \delta \widetilde{\ttc} -\epsilon_a \widetilde{\ttxi{a}} \delta \widetilde{\ttc},\\
\varpi^{\filt{3}}_{GR} & = \trintl{\filtint{3}} -\epsilon_n \delta\widetilde{\ttxi{n}} \delta \widetilde{\ttc} -\epsilon_m \delta\widetilde{\ttxi{m}} \delta \widetilde{\ttc} -\epsilon_a \delta\widetilde{\ttxi{a}} \delta \widetilde{\ttc},\\
S^{\filt{3}}_{GR} &= \trintl{\filtint{3}} \frac{1}{2}[\widetilde{\ttc},\widetilde{\ttc}]\epsilon_a \widetilde{\ttxi{a}}+\frac{1}{2}[\widetilde{\ttc},\widetilde{\ttc}]\epsilon_m \widetilde{\ttxi{m}}+\frac{1}{2}[\widetilde{\ttc},\widetilde{\ttc}]\epsilon_n \widetilde{\ttxi{n}};
\end{align}
\item The cohomological vector field $Q^{\filt{3}}_{GR}$
\begin{subequations}\label{Q_GR-vertex1}
\begin{align}
 Q^{\filt{3}}_{GR}\, {\widetilde{\ttc}} &= \frac{1}{2}[\widetilde{\ttc},\widetilde{\ttc}]\\
 Q^{\filt{3}}_{GR}\, {\widetilde{\ttxi{\mu}}} &= [\widetilde{\ttc}, \epsilon_a \widetilde{\ttxi{a}}]^{\comp{\mu}}+[\widetilde{\ttc}, \epsilon_m \widetilde{\ttxi{m}}]^{\comp{\mu}}+[\widetilde{\ttc}, \epsilon_n \widetilde{\ttxi{n}}]^{\comp{\mu}};
\end{align}
\end{subequations}
\item The projection to codimension-$3$ fields ${\pi}^{\filt{2}}_{GR} = \pi_{\mathrm{red}}^{\filt{3}} \circ {\check{\pi}}^{\filt{2}}$ is smooth, where $\check{\pi}^{\filt{2}}: \mathcal{F}^{\filt{2}}\to \check{\mathcal{F}}^{\filt{3}}$ is the restriction of the codimension-$2$ fields to the $3$-stratum $M^{\filt{3}}$ and $\pi_{\mathrm{red}}^{\filt{3}}: \check{\mathcal{F}}^{\filt{3}}\to \mathcal{F}^{\filt{3}}$ is the symplectic reduction \eqref{symplectic-reduction}.
\end{itemize}
\end{proposition}

\begin{remark}
Since $M^{\filt{3}}$ is a set of points, integration is here intended as sum over such points, with the space of fields on a single point being given by
\begin{equation*}
\wedge^2 V [1] \times \mathbb{R}^3 [1].
\end{equation*}
\end{remark}

\begin{proof}

In order to keep the notation light and readable we drop the tildes.  Once again, from the defining equation $\iota_{Q^{\filt{2}}} \varpi^{\filt{2}}= \delta S^{\filt{2}}+ \check{\pi}^{\filt{2}}\check{\alpha}^{\filt{3}}$ we get
$$\check{\alpha}^{\filt{3}}= \trintl{\filtint{3}} -\epsilon_n \xi^n \delta c -\epsilon_m \xi^m \delta c -e_a \xi^a \delta c $$
and
$$\check{\varpi}^{\filt{3}}= \delta \check{\alpha}^{\filt{3}}= \trintl{\filtint{3}} \epsilon_n \delta\xi^n \delta c +\epsilon_m \delta \xi^m \delta c -\delta e_a \xi^a \delta c +e_a \delta  \xi^a \delta c,$$
with $\check{\alpha}^{\filt{3}} \in \Omega^1(\check{\mathcal{F}}^{\filt{3}})$, the space of restrictions of fields in $\mathcal{F}^{\filt{2}}$ (and their normal jets) to the stratum $M^{\filt{3}}$, with $\check{\pi}^{\filt{2}}\colon \mathcal{F}^{\filt{2}} \longrightarrow \check{\mathcal{F}}^{\filt{3}}$.
There is only one nontrivial equation defining the kernel of $(\check{\varpi}^{\filt{3}})^\sharp$:
 $$ \epsilon_n X_{\xi^n} +\epsilon_m X_{\xi^m} +e_a X_{\xi^a} - X_{e_a}\xi^a=0.$$
This is solved as in the previous cases by
$$ X_{\xi}^{\mu}= (X_{e_a})^{\comp{\mu}} \xi^a$$
and, since this defines a subbundle of $T\check{\mathcal{F}}^{\filt{3}}$, it is possible to perform the symplectic reduction. {The kernel of $\check{\varpi}^{\filt{3}}$ is spanned by}
\begin{align*}
\mathbb{E}_a =& \X{e_a}\pard{}{e_a} + \X{e_a}^{\comp{\mu}}\xi^{a}\pard{}{\txi{\mu}},
\end{align*}
and ance again, flowing along $\mathbb{E}_a$, we are able to fix the vector $e_a$ to $\epsilon_a = (1+ \varepsilon'') e_a$ for some $\varepsilon'' >0$ {(a function on a finite set of points)} and consequently, $\xi^{a}[1] = (1+ \varepsilon'') \xi^a$, while the rest is left unchanged. This defines the symplectic reduction $\pi_{\mathrm{red}}^{\filt{3}} \colon \check{\mathcal{F}}^{\filt{3}}\longrightarrow \mathcal{F}^{\filt{3}}$, and the BV-BFV map $\pi^{\filt{2}}\coloneqq \pi_{\mathrm{red}}^{\filt{3}}\circ \check{\pi}^{\filt{2}}$. Then, the expression 
$$\alpha^{\filt{3}}= \trintl{\filtint{3}} -\epsilon_n \xi^n \delta c -\epsilon_m \xi^m \delta c -\epsilon_a \xi^a \delta c $$
is such that $\pi_{\mathrm{red}}^{\filt{3}*} \alpha^{\filt{3}}= \check{\alpha}^{\filt{3}}$.
Lastly, we compute the vertex action. With a calculation completely analogous to what was done in Propositions \ref{prop:BFVdata} and \ref{prop:BFFVdata} we compute $\check{S}^{\filt{3}}$ such that $\iota_{Q^{\filt{2}}}\iota_{Q^{\filt{2}}}\varpi^{\filt{2}}= 2 \check{\pi}^{\filt{2}*}\check{S}^{\filt{3}}$:
$$ \check{S}^{\filt{3}}= \trintl{\filtint{3}} \frac{1}{2}[c,c]e_a \xi^a+\frac{1}{2}[c,c]\epsilon_m \xi^m+\frac{1}{2}[c,c]\epsilon_n \xi^n$$
Then, we get that the vertex action 
$$ S^{\filt{3}}= \trintl{\filtint{3}} \frac{1}{2}[c,c]\epsilon_a \xi^a+\frac{1}{2}[c,c]\epsilon_m \xi^m+\frac{1}{2}[c,c]\epsilon_n \xi^n$$
satisfies $\check{S}^{\filt{3}}=\pi_{\mathrm{red}}^{\filt{3}*} S^{\filt{3}}$, and consequently
$$
\iota_{Q^{\filt{2}}}\iota_{Q^{\filt{2}}}\varpi^{\filt{2}}= 2 \pi^{\filt{2}*}S^{\filt{3}}.
$$
\end{proof}

\begin{corollary}[to the proof of Proposition \ref{prop:BFFFVdata}]\label{cor:BFFFVdata-projection}
The projection to codimension-$3$ fields, given a coordinate system adapted to the embedding $\filtint{3}\to\filtint{2}$ (see Remark \ref{rem:coordinatedependingquantities}):
\begin{equation}
\pi^{\filt{2}}_{GR}:
\begin{cases}
\widetilde{\ttc}:= \ttc \\
\widetilde{\ttxi{a}}:= (1+\varepsilon '')^{-1}\ttxi{a} \\
\widetilde{\ttxi{n}}:= \ttxi{n} \\
\widetilde{\ttxi{m}}:= \ttxi{m} \\
\epsilon_a:= (1 + \varepsilon '') \tte_a
\end{cases}
\end{equation}
where $\varepsilon ''\in \Gamma (M^{\filt{3}})$ and we used the notation explained in \eqref{restriction} and \eqref{vectorboundary} for the restriction of field and their normal components adapted as described in  Remark \ref{rem:notationhighercodimension}.
\end{corollary}

\section{BV-BFV equivalence} \label{sec:BF-GR_equivalence}
The goal of this section is to prove { Theorem \ref{thm:extended-equivalence}, based on Definition \ref{def:n-ext-strongBVeq}. We recall it here:}
\begin{theorem}
The fully extended BV-BFV theories $\mathfrak{F}^{\uparrow 3}_{GR}$ and $\mathfrak{F}^{\uparrow 3}_{\BFnd}$ are strongly equivalent.
\end{theorem}

\begin{remark}
Explicitly, we have to prove the existence of invertible symplectomorphisms $\psi^{\filt{k}}$ that make the following diagram commute.
\begin{center}
\begin{equation}\label{Commdiag}
\begin{tikzcd}[ row sep= 3 em, column sep= 4 em]
\mathcal{F}^{\filt{0}}_{GR} \arrow[r, "\pi^{\filt{0}}_{GR}"]\arrow[d, "\psi^{\filt{0}}"]& \mathcal{F}^{\filt{1}}_{GR} \arrow[r, "\pi^{\filt{1}}_{GR}"]\arrow[d, "\psi^{\filt{1}}"]& \mathcal{F}^{\filt{2}}_{GR}\arrow[r, "\pi^{\filt{2}}_{GR}"]\arrow[d, "\psi^{\filt{2}}"] & \mathcal{F}^{\filt{3}}_{GR}\arrow[d, "\psi^{\filt{3}}"] \\
\mathcal{F}^{\filt{0}}_{\BFnd} \arrow[r, "\pi^{\filt{0}}_{\BFnd}"] & 
\mathcal{F}^{\filt{1}}_{\BFnd} \arrow[r, "\pi^{\filt{1}}_{\BFnd}"] & 
\mathcal{F}^{\filt{2}}_{\BFnd}\arrow[r, "\pi^{\filt{2}}_{\BFnd}"] & 
\mathcal{F}^{\filt{3}}_{\BFnd} \\
\end{tikzcd}
\end{equation}
\end{center}
Note that the vertical symplectomorphisms preserve the action functionals, i.e. they satisfy $(\psi^{\filt{k}})^* S_{\BFnd}^{\filt{k}} = S_{GR}^{\filt{k}}$, and that the horizontal arrows on both lines have been already described in Section \ref{sec:BF} and \ref{sec:BVBFVGRproof} respectively. { The symplectomorphisms $\psi^{\filt{k}}$ are non-canonical, as they depend on the the choice of a basis in $\mathcal{V}$.}
\end{remark}

\subsection{Equivalence on the bulk}
A strong equivalence (see Definition \ref{def:strongBVeq}) between the BV data associated to BF theory and GR was proven in \cite[Theorem 10]{CSS2017}, provided that on $B$ is imposed a non degeneracy condition. We denote by $\BFnd$ nondegenerate BF theory (see Definition \ref{def:classBF}).  An explicit generating function for the canonical transformation between the two $(-1)$-symplectic spaces of fields has been given as well. We recall here the most important steps of the construction. 

Using the notation introduced in Subsections \ref{sec:BF} and \ref{sec:GR} the generating function\footnote{Some signs differ from the formula given in \cite[Theorem 10]{CSS2017}  because we are using a different convention for the signs in \eqref{GR-bulk-action}.} reads
\begin{equation} \label{Generatingfunction}
H = - B^\dag \left( e - \iota_\xi \omega^\dag + \frac{1}{2}\iota^2_\xi c^\dag \right) - \tau^\dag \left( -\iota_\xi e + \frac{1}{2} \iota^2_\xi \omega^\dag- \frac{1}{3}\iota^3_\xi c^\dag  \right) - A \omega^\dag + \chi c^\dag .
\end{equation}

Starting from this generating function,  we recover an explicit expression of the transformation $\psi^{\filt{0}} : \mathcal{F}_{GR}^{\filt{0}} \rightarrow \mathcal{F}_{\BFnd}^{\filt{0}}$. It can be found using the standard rules
\begin{equation} \label{Hamiltonequations}
p = - (-1)^{|q|} \pard{H}{q}; \qquad Q= (-1)^{|P|}\pard{H}{P}
\end{equation}

\noindent where $P = (\tau^\dag, B^\dag,  A, \chi)$, $Q=( \tau, B, A^\dag, \chi^\dag)$, $p=(\xi^\dag, e^\dag, \omega, c)$ and $q=( \xi, e, \omega^\dag, c^\dag)$. 
\begin{lemma} \label{lem:sympl_BV}
The symplectomorphism $\psi^{\filt{0}} \colon \mathcal{F}_{GR}^{\filt{0}} \rightarrow \mathcal{F}_{\BFnd}^{\filt{0}}$ is  given by
 \begin{equation}
\psi^{\filt{0}} \colon \begin{cases}
 B= e - \iota_\xi \omega^\dag + \frac{1}{2}\iota^2_\xi c^\dag &\\
 B^\dag= e^\dag - \iota_\xi \tau^\dag &\\
 A = \omega - \iota_\xi e^\dag + \frac{1}{2}\iota^2_\xi \tau^\dag &\\
  A^\dag = \omega^\dag &\\
 \chi = -c +\frac{1}{2}\iota^2_\xi e^\dag-\frac{1}{6}\iota^3_\xi \tau^\dag  &\\
 \chi^\dag = - c^\dag &\\
 \tau = -\iota_\xi e + \frac{1}{2} \iota^2_\xi \omega^\dag- \frac{1}{3}\iota^3_\xi c^\dag &\\
 p_a^\dag \tau^\dag= \xi^{\dag\comp{a}}_a - [e^\dag \omega_a^\dag]^{\comp{a}}+ [e^\dag \iota_\xi c_a^\dag]^{\comp{a}} &
\end{cases},
\end{equation}
whereas its inverse is given by
\begin{equation}
(\psi^{\filt{0}})^{ -1}: \begin{cases}
e= B + \iota_\xi A^\dag + \frac{1}{2}\iota^2_\xi \chi^\dag &\\
\xi^{a}= -\tau^{\comp{a}} - \tau^{\comp{b}}\tau^{\comp{c}} A_{bc}^{\dag(a)}+\tau^{\comp{b}}\tau^{\comp{c}}\tau^{\comp{d}} A_{cd}^{\dag\comp{e}} A_{eb}^{\dag\comp{a}}+\frac12 \tau^{\comp{b}}\tau^{\comp{c}}\tau^{\comp{d}}\chi_{bcd }^{\dag\comp{a}} &\\
\omega^\dag= A^\dag &\\
c^\dag= -\chi^\dag &\\
\omega = A+ \iota_{\xi}B^\dag + \frac{1}{2} \iota^2_\xi \tau^\dag&\\
c=-\chi+ \frac{1}{2}\iota^2_{\xi}B^\dag+ \frac{1}{3} \iota^3_\xi \tau^\dag&\\
 \xi^\dag_a= B^\dag(A_a^\dag+ \iota_\xi \chi_a^\dag) + \tau^\dag(B_a - \frac{1}{2}\iota^2_\xi \chi_a^\dag)&\\
e^\dag=B^\dag + \iota_{\xi}\tau^\dag
\end{cases},
\end{equation}
where the indices $\comp{a}$	 denote components with respect to the basis $\{B_a\}$ as defined in \eqref{componenttopform}.
\end{lemma}
\begin{proof}

We first apply \eqref{Hamiltonequations} to the generating function \eqref{Generatingfunction} and get
\begin{subequations}
\begin{align}
	B&= e - \iota_\xi \omega^\dag + \frac{1}{2}\iota^2_\xi c^\dag\\
	 A^\dag &=\omega^\dag \\
	\tau &=  -\iota_\xi e + \frac{1}{2} \iota^2_\xi \omega^\dag- \frac{1}{3}\iota^3_\xi c^\dag  \\
	\chi^\dag &= - c^\dag \\
	\xi^\dag_a &= B^\dag \omega_a^\dag - B^\dag \iota_\xi c_a^\dag + \tau^\dag e_a - \tau^\dag \iota_\xi \omega^\dag_a+ \tau^\dag \iota^2_\xi c^\dag_a  \label{explicitGR-BF_bulk5}\\
	 e^\dag&= B^\dag + \iota_\xi \tau^\dag \label{explicitGR-BF_bulk6}\\
	\omega &=\iota_\xi B^\dag + \frac{1}{2}\iota^2_\xi \tau^\dag +A \\
 	c &= \frac{1}{2}\iota^2_\xi B^\dag+\frac{1}{3}\iota^3_\xi \tau^\dag -\chi.
\end{align}
\end{subequations}
Equation \eqref{explicitGR-BF_bulk6} yields $ B^\dag= e^\dag - \iota_\xi \tau^\dag$ that in turn, inserted into \eqref{explicitGR-BF_bulk5}, gives
$$ \tau^\dag e_a= \xi^\dag_a - e^\dag \omega_a^\dag+ e^\dag \iota_\xi c_a^\dag .$$
Since this is an equation between objects valued in $\wedge^3 V$, we can extract an $e_a$ factor and get
$$p_a^\dag \tau^\dag= \xi^{\dag\comp{a}}_a - [e^\dag \omega_a^\dag]^{\comp{a}}+ [e^\dag \iota_\xi c_a^\dag]^{\comp{a}}.$$
Having $\tau^\dag$, we can now invert all other equations. 
An easy but lengthy computation shows that this symplectomorphism correctly satisfies $ \varpi_{GR}^{\filt{0}}= \psi^{\filt{0}*} \varpi_{\BFnd}^{\filt{0}}$ and it preserves the action functionals, as was shown in \cite{CSS2017}.
\end{proof} 
\begin{remark}
In order to build the inverse symplectomorphism $(\psi^{\filt{0}})^{-1}$ we must require that the image of $B$ be a basis of ${V}$ at every point. Thus the two theories are strongly equivalent only if also $B$ satisfies the non-degeneracy condition in Definition \ref{def:classBF}. {An extension of this equivalence to the degenerate case is presented in \cite{Szabo}, within the context of $L_\infty$ algebras.}
\end{remark}

\subsection{Equivalence on boundaries, corners, vertices}\label{Sect:BCV}
Since the form $\varpi_{GR}^{\filt{1}}$ of Equation \eqref{BoundarytwoformGR} is not in Darboux form, it is not possible to find a generating function. We can nonetheless produce an explicit symplectomorphism and its inverse:
\begin{equation}
\psi^{\filt{1}}: \begin{cases}
 B= \te - \iota_{\txi{}} \tom^\dag  &\\
 B^\dag= \te^\dag &\\
 A = \omega - \iota_{\txi{}} \te^\dag &\\
  A^\dag = \tom^\dag &\\
 \chi = -\tc +\frac{1}{2}\iota^2_{\txi{}} \te^\dag  &\\
 \tau = -\iota_{\txi{}} \te - \epsilon_n \txi{n} + \frac{1}{2} \iota^2_{\txi{}} \tom^\dag &\\
\end{cases}
(\psi^{\filt{1}})^{ -1}: \begin{cases}
\te= B+ \iota_{\txi{}}A^\dag &\\
\txi{a}= -\tau^{\comp{a}} - \tau^{\comp{a}}\tau^{\comp{b}} A_{ab}^{\dag\comp{a}}&\\
\txi{n}= -\tau^{\comp{n}} - \tau^{\comp{a}}\tau^{\comp{b}} A_{ab}^{\dag\comp{n}}&\\
\tom= A+ \iota_{\txi{}}B^\dag &\\
\tc=-\chi+ \frac{1}{2}\iota_{\txi{}}\iota_{\txi{}}B^\dag &\\
\tom^\dag=A^\dag  &\\
\te^\dag=B
^\dag 
\end{cases}
\end{equation}
where superscripts $\comp{n}$ and $\comp{a}$ denote the components with respect to $\{\epsilon_n,B_a\}$. It is straightforward to check that $\psi^{\filt{1}}\circ(\psi^{\filt{1}})^{ -1}=id$, $(\psi^{\filt{1}})^{ -1}\circ \psi^{\filt{1}}= id$ and $ \varpi^{\filt{1}}_{GR}= \psi^{\filt{1} *} \varpi^{\filt{1}}_{\BFnd}$. Analogously, on the corner (the codimension-$2$ stratum) we have the explicit transformation
\begin{equation}
\psi^{\filt{2}}: \begin{cases}
 B= \tte  &\\
 A = \ttom  &\\
 \chi = -\ttc &\\
 \tau = -\iota_{\ttxi{}} \tte - \epsilon_m \ttxi{m}- \epsilon_n \ttxi{n}. &\\
\end{cases}
(\psi^{\filt{2}})^{ -1}: \begin{cases}
 \tte =B  &\\
  \ttom = A  &\\
 \ttc = -\chi&\\
 \ttxi{m}= \tau^{\comp{m}} & \\
 \ttxi{n}= \tau^{\comp{n}} & \\ 
 \ttxi{a}= \tau^{\comp{a}},
\end{cases}
\end{equation}
while, on the vertex, we have
\begin{equation}
\psi^{\filt{3}}: \begin{cases}
 \chi = -\widetilde{\ttc} &\\
 \tau =- \epsilon_a \widetilde{\ttxi{a}} - \epsilon_m \widetilde{\ttxi{m}}- \epsilon_n \widetilde{\ttxi{n}} &\\
\end{cases}
\end{equation}
with inverses given by $\ttxi{m}= \tau^{\comp{m}}$, $\ttxi{n}= \tau^{\comp{n}}$ and $\ttxi{a}= \tau^{\comp{a}}$, i.e. the components of $\tau$ with respect to $\epsilon_m$, $\epsilon_n$ and $\epsilon_a$ respectively. Finally, it is straightforward to check that $(\psi^{\filt{k}})^* S_{\BFnd}^{\filt{k}} = S_{GR}^{\filt{k}}$ for $k=1,2,3$.

\subsection{Commutativity}
In this section we prove the commutativity of the three square subdiagrams of the diagram \eqref{Commdiag}. This is sufficent to prove  commutativity as a  whole. For the sake of clarity, we denote the BF variables on the $1$-stratum (and subsequent $2$- and $3$-strata) with a tilde, analogously to the GR notation. 
To avoid confusion, we explicitly denote the restriction to the $1$-stratum (resp. $2$- and $3$-stratum) with an apex ', e.g. $e'\equiv e\vert_{\filtint{1}}$ and $\te{}' \equiv \te\vert_{\filtint{2}}$.
\begin{proposition}
Diagram \eqref{Commdiag} is commutative.
\end{proposition}
\begin{proof}
The first square is
\begin{center}
\begin{tikzcd}[ row sep= 3 em, column sep= 4 em]
\mathcal{F}_{GR}^{\filt{0}} \arrow[r, "\pi^{\filt{0}}_{GR}"]\arrow[d, "\psi^{\filt{0}}"]& \mathcal{F}^{\filt{1}}_{GR}  \arrow[d, "\psi^{\filt{1}}"] \\
\mathcal{F}_{\BFnd}^{\filt{0}} \arrow[r, "\pi^{\filt{0}}_{\BFnd}"]& 
\mathcal{F}^{\filt{1}}_{\BFnd}
\end{tikzcd}
\end{center}
 The left-bottom composition $\pi^{\filt{0}}_{\BFnd}\circ \psi^{\filt{0}}$ reads
\begin{align*}
 \widetilde{B}&=( e - \iota_\xi \omega^\dag + \frac{1}{2}\iota^2_\xi c^\dag )'=e' - \iota_{\xi '} \omega^{\dag '} - \omega_n^{\dag '}\xi^{n'} +  \iota_{\xi '}c_n^{\dag '}\xi^{n'}  \\
\widetilde{B}^\dag&= (e^\dag - \iota_\xi \tau^\dag)'= e^{\dag '}- \tau_n^{\dag '}\xi^{n'}\\
\widetilde{A} &= (\omega - \iota_\xi e^\dag + \frac{1}{2}\iota^2_\xi \tau^\dag)'=\omega' - \iota_{\xi '} e^{\dag '}-e_n^{\dag '}\xi^{n'}  + \iota_{\xi '} \tau_n^{\dag '} \xi^{n'}=\omega' - \iota_{\xi '}\widetilde{B}^\dag -e_n^{\dag '}\xi^{n'} \\
\widetilde{A}^\dag &= \omega^{\dag '} \\
\widetilde{\chi} &= (-c +\frac{1}{2}\iota^2_\xi e^\dag-\frac{1}{6}\iota^3_\xi \tau^\dag)'= -c' +\frac{1}{2}\iota^2_{\xi '}  e^{\dag '}  + \iota_{\xi '} e_n^{\dag '} \xi^{n'}-   \frac{1}{2}\iota^2_{\xi '} \tau_n^{\dag '} \xi^{n'} \\
	&= -c' +\frac{1}{2}\iota^2_{\xi '}\widetilde{B}^\dag  + \iota_{\xi '} e_n^{\dag '} \xi^{n'}\\
\widetilde{\tau} &=( -\iota_\xi e + \frac{1}{2} \iota^2_\xi \omega^\dag- \frac{1}{3}\iota^3_\xi c^\dag)' = -\iota_{\xi '} e' -e'_n\xi^{n'}+ \frac{1}{2} \iota^2_{\xi '} \omega^{\dag '}+ \iota_{\xi '}\omega_n^{\dag '}\xi^{n'}-\frac{1}{2} \iota^2_{\xi '}c_n^{\dag '}\xi^{n'}
\end{align*}
where
\begin{align*}
\widetilde{\tau}_n^\dag=& (\chi_a\mathsf{V}_n^{\comp{a}} - [e_n^\dag \omega_a^\dag- e^\dag \omega_{an}^\dag]^{\comp{a}}+ [e^\dag \iota_\xi c_{an}^\dag]^{\comp{a}}+\chi_n\mathsf{V}_n^{\comp{n}} - [e_n^\dag \omega_n^\dag]^{\comp{n}}+ [e_n^\dag \iota_\xi c_n^\dag]^{\comp{n}})'.
\end{align*}
The top-right composition $\psi^{\filt{1}}\circ \pi^{\filt{0}}_{GR}$ reads:
\begin{align*}
 \widetilde{B}&= \te - \iota_{\txi{}} \tom^\dag  =e' - \omega^{\dag '}_n\xi^{n'} +\iota_{\xi '} c^{\dag '}_n\xi^{n '} - \iota_{\xi '}\omega^{\dag '} \\
 \widetilde{B}^\dag&= \te^\dag =e^{\dag} - \chi_a\xi^n\mathsf{V}^{\comp{a}}  +\left(\omega_{na}^\dag\xi^n e^\dag \right)^{\comp{a}}- (e^\dag_n\xi^n\omega^\dag_a)^{\comp{a}} - \left(\iota_\xi c^\dag_{na}\xi^n e^\dag\right)^{\comp{a}} \\
& - \chi_n\xi^n\mathsf{V}^{\comp{n}}  + \left(\omega_n^\dag\xi^n e^\dag_n \right)^{\comp{n}}- (e^\dag_n \iota_\xi c^\dag_n\xi^n)^{\comp{n}} \\
 \widetilde{A} &= \tom - \iota_{\txi{}} \te^\dag= \omega - e^\dag_n\xi^n -\iota_{\xi} \te^\dag \\
 \widetilde{A}^\dag &= \tom^\dag = \omega^{\dag '} \\
 \widetilde{\chi} & = -\tc +\frac{1}{2}\iota^2_{\txi{}} \te^\dag=  c -  \iota_\xi e^\dag_n\xi^n +\frac{1}{2}\iota^2_{\xi} \te^\dag \\
 \widetilde{\tau} &= -\iota_{\txi{}} \te - \epsilon_n \txi{n} + \frac{1}{2} \iota^2_{\txi{}} \tom^\dag = - \iota_{\xi '} e' + \iota_{\xi '}\omega^{\dag '}_n\xi^{n'} -\iota^2_{\xi '} c^{\dag '}_n\xi^{n'} -e'_n \xi^{n'} + \frac{1}{2} \iota^2_{\xi '} \omega^{\dag '},
\end{align*}
where we used that $e'_n \xi^{n'} = \epsilon_n \txi{n} $.
The rows of $\widetilde{B}$, $\widetilde{A}^\dag$ and $ \widetilde{\tau}$ coincide in both cases. The expressions of $ \widetilde{B}^\dag$ coincide as well, hence also the rows of $ \widetilde{A}$ and $ \widetilde{\chi}$ give the same result.
The second square is:
\begin{center}
\begin{tikzcd}[ row sep= 3 em, column sep= 4 em]
\mathcal{F}^{\filt{1}}_{GR} \arrow[r, "\pi^{\filt{1}}_{GR}"]\arrow[d, "\psi^{\filt{1}}"]& \mathcal{F}^{\filt{2}}_{GR}\arrow[d, "\psi^{\filt{2}}"]  \\
 \mathcal{F}^{\filt{1}}_{\BFnd} \arrow[r, "\pi^{\filt{1}}_{\BFnd}"]& \mathcal{F}^{\filt{2}}_{\BFnd}\\
\end{tikzcd}
\end{center}
 The left-bottom composition $\pi^{\filt{1}}_{\BFnd} \circ \psi^{\filt{1}}$ is
\begin{align*}
 \widetilde{\widetilde{B}}&= (\te - \iota_{\txi{}} \tom^\dag)'  =\te ' -  \tom_m^{\dag '}\txi{m'}  \\
\widetilde{\widetilde{ A}} &= (\tom - \iota_{\txi{}} \te^\dag)' = \tom '  -  \te_m^{\dag '}\txi{m'}\\
\widetilde{\widetilde{ \chi}}& = (-\tc +\frac{1}{2}\iota^2_{\txi{}} \te^\dag)' =  -\tc '+  \iota_{\txi{'}}\te_m^{\dag '}\txi{m'} \\
\widetilde{\widetilde{ \tau}} &= (-\iota_{\txi{}} \te - \epsilon_n \txi{n} + \frac{1}{2} \iota^2_{\txi{}} \tom^\dag)'=-\iota_{\txi{'}} \te ' - \epsilon_n \txi{n '}- \te '_m \txi{m '} +  \iota_{\txi{'}}  \tom_m^{\dag '}\txi{m'}  \\
\end{align*}
 while for the top-right composition $\psi^{\filt{2}}\circ \pi^{\filt{1}}_{GR}$ we have
\begin{align*}
\widetilde{\widetilde{B}}&= \tte = \te '-\tom^{\dag '}_m \txi{m'} \\
\widetilde{\widetilde{ A}} &= \ttom = \tom '-\te^{\dag '}_m \txi{m'}\\
\widetilde{\widetilde{ \chi}} &= -\ttc = -\tc '+ \iota_{\txi{'}} \te^{\dag '}_m \txi{m'}\\
\widetilde{\widetilde{ \tau}} &= -\iota_{\ttxi{}} \tte - \epsilon_m \ttxi{m}- \epsilon_n \ttxi{n}=-\iota_{\txi{'}} (\te '-\tom^{\dag '}_m \txi{m'}) - \te_m ' \txi{m'}- \epsilon_n \txi{n'}, 
\end{align*}
(again we used $\te '_m \txi{m'} = \epsilon_m \ttxi{m}$) and the expressions are identical.
The last square subdiagram is
\begin{center}
\begin{tikzcd}[ row sep= 3 em, column sep= 4 em]
\mathcal{F}^{\filt{2}}_{GR}\arrow[r, "\pi^{\filt{2}}_{GR}"]\arrow[d, "\psi^{\filt{2}}"] & \mathcal{F}^{\filt{3}}_{GR}\arrow[d, "\psi^{\filt{3}}"] \\
\mathcal{F}^{\filt{2}}_{\BFnd}\arrow[r, "\pi^{\filt{2}}_{\BFnd}"] & \mathcal{F}^{\filt{3}}_{\BFnd} \\
\end{tikzcd}
\end{center}
The two compositions are $\pi^{\filt{2}}_{\BFnd}\circ \psi^{\filt{2}}$:
\begin{align*}
\widetilde{\widetilde{\widetilde{ \chi}}} &= -\ttc' \\
\widetilde{\widetilde{\widetilde{ \tau}}} &= (-\iota_{\ttxi{}} \tte - \epsilon_m \ttxi{m}- \epsilon_n \ttxi{n})'=-\ \tte'_a \ttxi{a '}- \epsilon_m \ttxi{m'}- \epsilon_n \ttxi{n'}
\end{align*}
and $\psi^{\filt{3}}\circ \pi^{\filt{2}}_{GR}$:
\begin{align*}
\widetilde{\widetilde{\widetilde{ \chi}}} &= -\widetilde{\ttc} = - \ttc' \\
\widetilde{\widetilde{\widetilde{ \tau}}} &=- \epsilon_a \widetilde{\ttxi{a}} - \epsilon_m \widetilde{\ttxi{m}}- \epsilon_n \widetilde{\ttxi{n}}= - \tte_a '\ttxi{a '}- \epsilon_m \ttxi{m'}- \epsilon_n \ttxi{n'}
\end{align*}
using once again $\tte '_a \txi{a'} = \epsilon_a \ttxi{a}$.
\end{proof}

\section{Cosmological constant}
In this section we consider BF theory and GR theory with the addition of what is generally known as the cosmological term: a cubic term in $B$ (respectively $e$). Classically this amounts to considering the functionals
\begin{equation*}
S^{cl}_{\Lambda BF} = \trintl{M} B\wedge F_A + \frac{1}{6}\Lambda B \wedge B \wedge B, \qquad S^{cl}_{\Lambda GR} = \trintl{M} e\wedge F_\omega  +\frac{1}{6}\Lambda e \wedge e \wedge e
\end{equation*}
where $\Lambda \in \mathbb{R}$ is a constant. The corresponding BV expression are (\cite{CSS2017})
\begin{equation*}
S_{\Lambda BF}=S_{BF}+\frac{1}{6} \trintl{M} \Lambda \mathcal{B} \wedge \mathcal{B} \wedge \mathcal{B}, \qquad S_{GR}=S_{\Lambda GR}+ \frac{1}{6}\trintl{M} \Lambda e \wedge e \wedge e
\end{equation*}
BF theory with this additional term is still fully extendable and self-similar, so, using the notation of Theorem \ref{thm:BF-BV_BFV} we get

\begin{equation*}
S^{\filtBF{k}}_{\Lambda BF}=S_{BF}+\frac{1}{6} \trintl{ \filtintBF{k}} \Lambda \mathcal{B} \wedge \mathcal{B} \wedge \mathcal{B}.
\end{equation*}

Since the additional cosmological term does not contain any derivative, also GR is fully extendable and the reductions are not modified. The actions in higher codimensions are
\begin{align*}
S^{\filt{1}}_{\Lambda GR} &= S^{\filt{1}}_{GR} - \frac{1}{2} \trintl{\filtint{1}} \Lambda \epsilon_n \txi{n} \te \te    \\
S^{\filt{2}}_{\Lambda GR} &= S^{\filt{2}}_{GR} + \trintl{\filtint{2}} \Lambda \epsilon_n \ttxi{n} \epsilon_m \ttxi{m} \tte    \\
S^{\filt{3}}_{\Lambda GR} &= S^{\filt{3}}_{GR} - \trintl{\filtint{3}} \Lambda \epsilon_n \widetilde{\ttxi{n}} \epsilon_m \widetilde{\ttxi{m}} \epsilon_a \widetilde{\ttxi{a}}    .
\end{align*}
The two fully extended theories are still strongly equivalent and the map realizing the equivalence (namely the ones appearing in the diagram \eqref{Commdiag}) remain unchanged. We have just to check that the actions are still preserved by the corresponding symplectomorphisms. In the bulk this has been proved in \cite[section 2.3]{CSS2017}. The same argument can be adapted to higher codimension actions. The equations
$(\psi^{\filt{k}})^* S_{\Lambda \BFnd}^{\filt{k}} =  S_{\Lambda GR}^{\filt{k}}$ for $k=0, \dots 3$ can also be verified by direct computation.

\printbibliography
\end{document}